\documentclass[reprint,aps,twocolumn,superscriptaddress,showpacs,notitlepage,10pt]{revtex4-2}

\usepackage{amsmath,amssymb,amsfonts,amstext,amsthm}
\usepackage{graphicx,breakurl,color,setspace}
\usepackage{mathtools}
\usepackage{dsfont}
\usepackage[T1]{fontenc}
\usepackage[latin9]{inputenc}
\usepackage[normalem]{ulem}
\usepackage[unicode=true,pdfusetitle,pdfborder={0 0 1}]{hyperref}
\usepackage{xcolor}

\hypersetup{
    colorlinks,
    linkcolor=[RGB]{243 68 110},
    citecolor=[RGB]{0 160 63},
    urlcolor=[RGB]{57 115 239}
}

\usepackage{upgreek}
\usepackage{bm}
\usepackage{mathrsfs}

\theoremstyle{definition}
\newtheorem{theorem}{Theorem}

\newtheorem{definition}[theorem]{Definition}
\newtheorem{proposition}[theorem]{Proposition}
\newtheorem{lemma}[theorem]{Lemma}

\begin{document}

\title{Magic States: Cheaper than Lattice Surgery CNOT gates} 
\title{Efficient Magic State Cultivation on $\mathbb{RP}^2$}
\author{Zi-Han Chen}

\email{czh007@mail.ustc.edu.cn}

\affiliation{
Hefei National Research Center for Physical Sciences at the Microscale and School of Physical Sciences,
University of Science and Technology of China, Hefei 230026, China}
\affiliation{
Shanghai Research Center for Quantum Science and CAS Center for Excellence in Quantum Information and Quantum Physics,
University of Science and Technology of China, Shanghai 201315, China}
\affiliation{
Hefei National Laboratory, University of Science and Technology of China, Hefei 230088, China}

\author{Ming-Cheng Chen}

\affiliation{
Hefei National Research Center for Physical Sciences at the Microscale and School of Physical Sciences,
University of Science and Technology of China, Hefei 230026, China}
\affiliation{
Shanghai Research Center for Quantum Science and CAS Center for Excellence in Quantum Information and Quantum Physics,
University of Science and Technology of China, Shanghai 201315, China}
\affiliation{
Hefei National Laboratory, University of Science and Technology of China, Hefei 230088, China}

\author{Chao-Yang Lu}

\affiliation{
Hefei National Research Center for Physical Sciences at the Microscale and School of Physical Sciences,
University of Science and Technology of China, Hefei 230026, China}
\affiliation{
Shanghai Research Center for Quantum Science and CAS Center for Excellence in Quantum Information and Quantum Physics,
University of Science and Technology of China, Shanghai 201315, China}
\affiliation{
Hefei National Laboratory, University of Science and Technology of China, Hefei 230088, China}

\author{Jian-Wei Pan}

\affiliation{
Hefei National Research Center for Physical Sciences at the Microscale and School of Physical Sciences,
University of Science and Technology of China, Hefei 230026, China}
\affiliation{
Shanghai Research Center for Quantum Science and CAS Center for Excellence in Quantum Information and Quantum Physics,
University of Science and Technology of China, Shanghai 201315, China}
\affiliation{
Hefei National Laboratory, University of Science and Technology of China, Hefei 230088, China}

\date{\today}

\begin{abstract}
Preparing high-fidelity logical magic states is crucial for fault-tolerant quantum computation. Among prior attempts to reduce the substantial cost of magic state preparation, magic state cultivation (MSC)~\cite{gidney_magic_2024}, a recently proposed protocol for preparing $\mathrm{T}$ states without magic state distillation, achieves state-of-the-art efficiency. Inspired by this work~\cite{gidney_magic_2024}, we propose a new MSC procedure that would produce a logical $\mathrm{T}$ state on a rotated surface code at a further reduced cost. For our MSC protocol, we define a new code family, the $\mathbb
{RP}^2$ code, by putting the rotated surface code on $\mathbb{RP}^2$ (a two-dimensional manifold), as well as two self-dual CSS codes named SRP-3 and SRP-5 respectively. Small $\mathbb{RP}^2$ codes are used to hold logical information and checked by syndrome extraction (SE) circuits. We design fast morphing circuits that enable switching between a distance 3 (5) $\mathbb{RP}^2$ code and an SRP-3 (SRP-5) code on which we can efficiently check the correctness of the logical state. To preserve the high accuracy of the cultivated logical $\mathrm{T}$ state, we design an efficient and easy-to-decode expansion stage that grows a small $\mathbb{RP}^2$ code to a large rotated surface code in one round. Our MSC protocol utilizes non-local connectivity, available on both neutral atom array and ion trap platforms.  According to our Monte Carlo sampling results, our MSC protocol requires about an order of magnitude smaller space-time volume to reach a target logical error rate around $10^{-9}$ compared to the original MSC protocol. 
\end{abstract}

\maketitle		

Fault-tolerant quantum computation (FTQC) is necessary, given noisy physical operations, for implementing quantum algorithms with arbitrarily high precision. The rotated surface code~\cite{bravyi_quantum_1998,dennis_topological_2002,horsman_surface_2013}, with high circuit noise threshold $\sim 1\%$, is promising for practical implementation of FTQC and admits experimental demonstration definitively below threshold~\cite{google_quantum_ai_and_collaborators_quantum_2025}. Tremendous theoretical progress on surface code-based logical processing has been made over the last two decades and includes constructing and optimizing logical Clifford operations~\cite{fowler_surface_2012,horsman_surface_2013,litinski_game_2019,bombin_logical_2023,beverland_fault_2024,geher_error-corrected_hadamard_2024,brown_poking_2017,gidney_inplace_2024,moussa_transversal_2016,chen_transversal_2024} and magic state preparations~\cite{bravyi_universal_2005,fowler_bridge_2012,fowler_surface_2012,li_2015,litinski_game_2019,fowler_low_2019,gidney_efficient_ccz_catalyzed_2019,gidney_flexible_autoccz_2019,itogawa_even_2024,hirano_leveraging_2024} (magic states are used to implement non-Clifford gates via gate teleportation~\cite{gottesman_demonstrating_1999,zhou_methodology_2000}). The substantial cost of preparing magic states (with low logical error rates under a practically relevant physical noise rate $\sim 10^{-3}$) has long persisted~\cite{fowler_bridge_2012,fowler_surface_2012,litinski_game_2019,fowler_low_2019,gidney_efficient_ccz_catalyzed_2019,gidney_flexible_autoccz_2019,itogawa_even_2024,hirano_leveraging_2024,vasmer_2019_3d_surface_code,chamberland_very_2020,heussen2024efficientfaulttolerantcodeswitching,lee_low-overhead_2024} until the recent proposal of  a new $\mathrm{T}$ state preparation protocol, magic state cultivation (MSC)~\cite{gidney_magic_2024}, that can produce high-fidelity $\mathrm{T}$ states with roughly the same space-time cost as performing lattice surgery CNOT gates of the same fidelity.    

The MSC procedure in~\cite{gidney_magic_2024} cultivates a high-fidelity $\mathrm{T}$ state on a small distance 3 or distance 5 color code, which is then ``grafted'' to a larger matchable code (essentially a rotated surface code deformed at a corner). While the cultivation process (on small color codes) is highly efficient thanks to a new technique, logical double-checking, developed in~\cite{gidney_magic_2024}, the ``grafting'' process (also novel) incurs extra post-selection and syndrome extraction (SE) rounds. Furthermore, for physical systems such as neutral atom array~\cite{bluvstein_quantum_2022} and ion trap~\cite{moses_race_track_2023} platforms that support non-local connectivity, it would be more convenient to prepare $\mathrm{T}$ states on regular surface codes, since these $\mathrm{T}$ states can be teleported~\cite{gottesman_demonstrating_1999,zhou_methodology_2000} more efficiently via transversal CNOTs (instead of lattice surgery CNOTs) to apply logical $\mathrm{T}$ gates on other regular surface codes. Recently, another MSC procedure~\cite{vaknin_magic_2025} is proposed and has a more efficient expansion stage, which grows a smaller surface code to a larger surface code. However, this procedure does not have an efficient cultivation process as the original MSC procedure and overall is less efficient than the original MSC procedure. 

\begin{figure}
    \centering
    \includegraphics[width=\columnwidth]{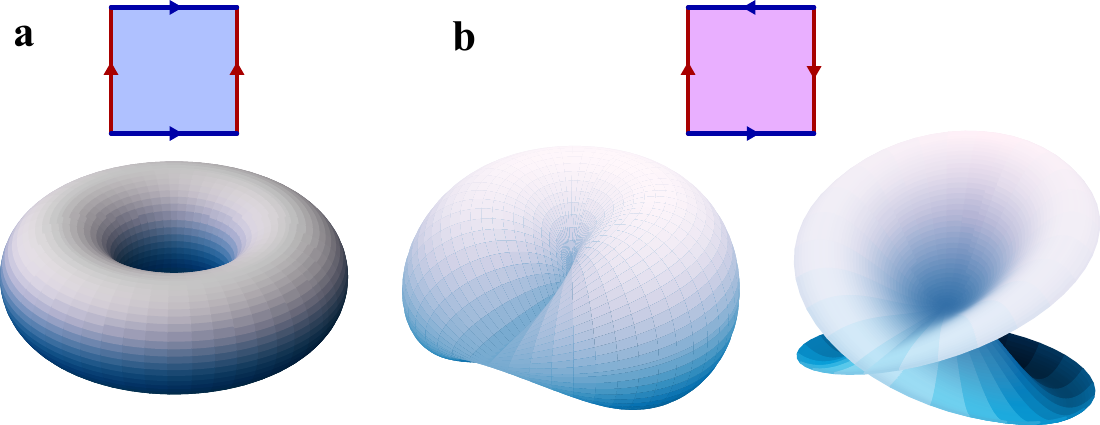}
    \caption{Torus and $\mathbb{RP}^2$ obtained by identifying the sides of a square. Red (Blue) arrows on the sides of a square patch are identified. (\textbf{a}) Illustration of the torus. (\textbf{b}) Illustration of the $\mathbb{RP}^2$. The shape (representing the $\mathbb{RP}^2$) in the middle is obtained from the shape on the right by gluing together its upper and lower rim.}
    \label{fig: torus and rp2}
\end{figure}

In this work, we propose a new MSC protocol that simultaneously has an efficient cultivation process and an efficient and easy-to-decode expansion and stabilization (ES) process. Moreover, our MSC protocol uses non-local connectivity and produces $\mathrm{T}$ states on regular rotated surface codes. For our MSC protocol, we design a new code family, the $\mathbb{RP}^2$ code, and two self-dual CSS codes, named SRP-3 and SRP-5 respectively.  The $\mathbb{RP}^2$ code is a variant of the rotated surface code and is constructed by placing the rotated surface code on $\mathbb{RP}^2$. See Fig.~\ref{fig: torus and rp2} for an illustration of the torus, on which the toric code~\cite{kitaev_fault-tolerant_2003} (the predecessor of the surface code) is defined, and the $\mathbb{RP}^2$. For our MSC protocol, we use the $\mathbb{RP}^2$-3 code (with distance 3) and the $\mathbb{RP}^2$-5 code (with distance 5) in the $\mathbb{RP}^2$ code family. Due to their similarity to rotated surface codes, we design an efficient expansion circuit that grows an $\mathbb{RP}^2$-3 ($\mathbb{RP}^2$-5) code to a larger regular rotated surface code in a single round. To also obtain an efficient cultivation process that grows a high-fidelity $\mathrm{T}$ state on an $\mathbb{RP}^2$-3 ($\mathbb{RP}^2$-5) code, we construct fast morphing circuits that morphs the $\mathbb{RP}^2$-3 ($\mathbb{RP}^2$-5) code to an SRP-3 (SRP-5) code, on which we can implement logical double-checking similar to~\cite{gidney_magic_2024}, and then back. According to Monte Carlo sampling results, our MSC protocol requires nearly an order of magnitude smaller space-time cost to reach a logical error rate around $1.5\times 10^{-6}$ or $1\times10^{-9}$ than the original MSC protocol~\cite{gidney_magic_2024}. We overview our MSC protocol and results in the following section (Sec.~\ref{sec: results}), at the end of which the structure of the rest of the paper is laid out. 

\section{Results}\label{sec: results}
\begin{figure}
    \centering
    \includegraphics[width=\columnwidth]{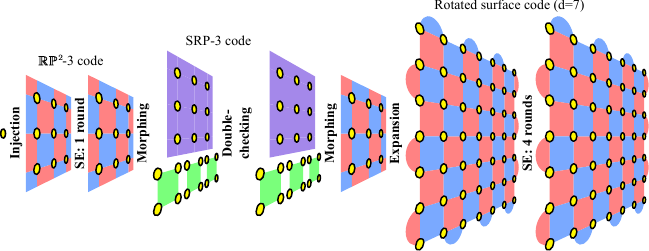}
    \caption{The MSC-3 protocol. Data qubits are marked by yellow circles. Red (blue) tiles represent Pauli $\mathrm{X}$ ($\mathrm{Z}$) stabilizers. Each stage of the protocol is illustrated between two code patches. A few stabilizers of the $[[15,1,3]]$ SRP-3 code are illustrated by green squares, each of which represents an $\mathrm{X}$ stabilizer \textit{and} a $\mathrm{Z}$ stabilizer. Not all stabilizers of the SRP-3 code are illustrated. Also, within each slice of SRP-3 code patches, some of the data qubits on the green squares are identified with qubits on the purple patch. The final SE round on the rotated surface code is implemented by a perfect SE circuit followed by a perfect logical readout.}
    \label{fig: MSC-3 procedure}
\end{figure}

\begin{figure*}
    \centering
    \includegraphics[width=\textwidth]{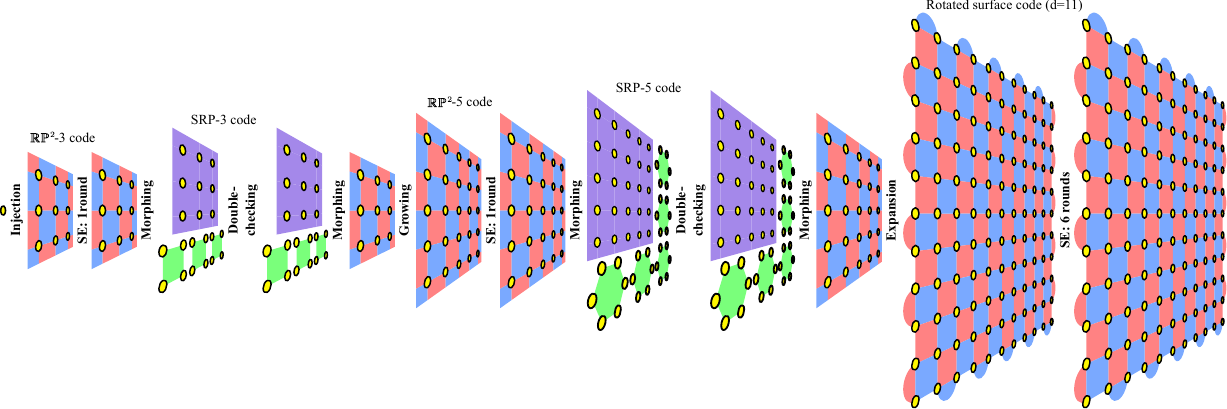}
    \caption{The MSC-5 protocol. Illustration convention follows from Fig.~\ref{fig: MSC-3 procedure}. A few stabilizers of the $[[35,1,5]]$ SRP-5 code are illustrated by green hexagons, each of which represents an $\mathrm{X}$ stabilizer \textit{and} a $\mathrm{Z}$ stabilizer. Not all stabilizers of the SRP-5 code are illustrated. Also, within each slice of SRP-5 code patches, some of the data qubits on the green hexagons are identified with qubits on the purple patch. The final SE round on the rotated surface code is implemented by a perfect SE circuit followed by a perfect logical readout.}
    \label{fig: MSC-5 procedure}
\end{figure*}

Our MSC protocol has two versions, a smaller version named MSC-3 (Fig.~\ref{fig: MSC-3 procedure}) and a full-sized version named MSC-5 that can reach lower logical error rates (Fig.~\ref{fig: MSC-5 procedure}). Both versions contain two processes, a cultivation process for producing a high-fidelity T state on a smaller code and an ES process for transferring the logical qubit from the smaller code to a larger rotated surface code. During the cultivation process of the MSC-3 (MSC-5) procedure, a single qubit $\mathrm{T}$ state is grown (gradually grown) to an $\mathbb{RP}^2$-3 ($\mathbb{RP}^2$-5) code with distance $3$ ($5$). SE rounds as well as logical double-checks are applied at the intervals of code growths to improve the fault distance as in~\cite{gidney_magic_2024}. Interestingly, these SE rounds are only performed on $\mathbb{RP}^2$ codes and the logical double-checks are only performed on the SRP-3 or the SRP-5 code. Each MSC attempt would be discarded if any detector during the cultivation process is triggered. During the ES process of our MSC-3 (MSC-5) protocol, we first expand the $\mathbb{RP}^2$-3 ($\mathbb{RP}^2$-5) code to a distance 7 (distance 11) rotated surface code; then apply a few SE rounds on the rotated surface code (Fig.~\ref{fig: MSC-3 procedure},~\ref{fig: MSC-5 procedure}). Similar to~\cite{gidney_magic_2024,vaknin_magic_2025}, at the end of the ES process, another round of post-selection, which we call the final-stage post-selection, is implemented to discard MSC attempts that survive the cultivation process but may contain troublesome errors which would significantly degrade the logical fidelity. In~\cite{gidney_magic_2024,vaknin_magic_2025}, a soft output named complementary gap is calculated for each decoding shot and is used to decide whether the shot should be discarded for the final-stage post-selection. In our work, we adopt a different method proposed in~\cite{meister_efficient_2024} for obtaining soft outputs. In particular, we extract a \textit{two-dimensional soft output} $(\phi_{\mathbb{RP}^2},\phi_{bd})\in\mathbb{Z}_{\geq 0}^2$ for each shot and set cutoffs on both $\phi_{\mathbb{RP}^2}$ and $\phi_{bd}$ for final-stage post-selection. Exact simulation of our MSC protocols is inefficient due to the presence of $\mathrm{T}$ and $\mathrm{T}^{\dagger}$ gates. Following~\cite{gidney_magic_2024}, we instead replace all $\mathrm{T}$ ($\mathrm{T}^{\dagger}$) gates in our protocols with $\mathrm{S}$ ($\mathrm{S}^{\dagger}$) gates and perform Monte Carlo sampling on such altered circuits (cultivating $\mathrm{S}|+\rangle$ states instead of $\mathrm{T}$ states) under uniform depolarizing circuit noise~\cite{gidney_magic_2024} with noise rate $p=0.001$. Thanks to our highly efficient ES process, our MSC protocol has higher end-to-end survival rates (Fig.~\ref{fig: MSC results}(\textbf{a}) and Tab.~\ref{tab: MSC results comparison}) and requires a smaller code size at the end than the original MSC protocol~\cite{gidney_magic_2024} to achieve the same logical error rates. As shown in Fig.~\ref{fig: MSC results}(\textbf{b}), our MSC-3 (MSC-5) protocol requires nearly an order of magnitude smaller space-time cost to achieve a logical error rate around $1.5\times 10^{-6}$ ($1\times10^{-9}$) compared to the original MSC protocol~\cite{gidney_magic_2024}. 

\begin{figure*}
    \centering
    \includegraphics[width=\textwidth]{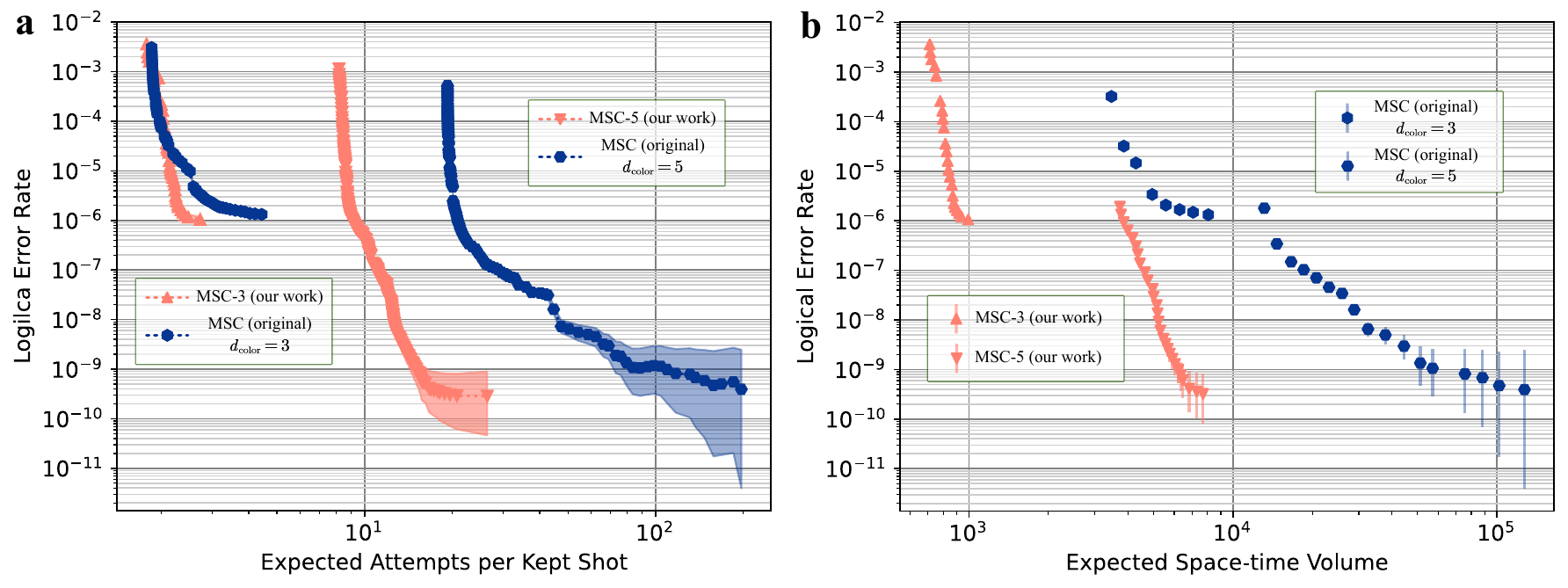}
    \caption{Monte Carlo sampling results under uniform depolarizing circuit noise~\cite{gidney_magic_2024} with noise rate $p=0.001$ for our MSC protocol and the original MSC protocol (taken directly from~\cite{gidney_magic_2024}). Here $d_{\text{color}}$ is the distance of the color code right before the ``grafting'' process in the original MSC protocol~\cite{gidney_magic_2024}. (\textbf{a}) Trade-offs between the logical error rate and the expected number of attempts per kept shot. (\textbf{b}) Expected space-time cost to reach a certain logical error rate. Logical error rates are obtained via maximum likelihood estimation over collected samples~\cite{gidney_sinter}. Error bars and shaded regions both indicate logical error rate hypotheses whose likelihood is within a factor 1000 of the maximum likelihood. }
    \label{fig: MSC results}
\end{figure*}

\begin{table*}
    \centering
    \begin{tabular}{||c|c|c|c|c|c|c||}
         \hline
         \hline
         MSC protocol &  Logical Error rate & Discard Rate & Logical Error rate & Discard Rate & Final distance & Footprint \\
                      &      (Ungrown)      & (Ungrown)    &    (End-to-end)    & (End-to-end) & (End-to-end) & (End-to-end) \\
         \hline
         MSC-3 (this work)& $8.3\times 10^{-7}$ & $48.6\%$ & $1.5\times 10^{-6}$ & $58\%$ & 7 & 103 qubits \\
         \hline
         MSC ($d_{\text{color}}=3$)~\cite{gidney_magic_2024} & $3.4\times 10^{-7}$ & $31.3\%$ & $1.5\times 10^{-6}$ & $74\%$ & 15 & 454 qubits \\
         \hline
         MSC-5 (this work) & $3.6\times 10^{-10}$ & $90.5\%$ & $1\times 10^{-9}$ & $93.4\%$ & 11 & 251 qubits \\
         \hline
         MSC ($d_{\text{color}}=5$)~\cite{gidney_magic_2024} & $3.5\times 10^{-10}$ & $85.6\%$ & $1\times 10^{-9}$ & $98.9\%$ & 15 & 463 qubits \\
         \hline
         \hline
    \end{tabular}
    \caption{Selected Monte Carlo sampling results under uniform depolarizing circuit noise~\cite{gidney_magic_2024} with noise rate $p=0.001$. Here $d_{\text{color}}$ is the distance of the color code right before the ``grafting'' process in the original MSC protocol~\cite{gidney_magic_2024}. Performance of each cultivation process alone (which only produces $\mathrm{T}$ states on a small code) is labeled by `Ungrown'. Benchmarking the cultivation process alone is implemented by sampling its Clifford variant, which cultivates $\mathrm{S}|+\rangle$ on the small code, appended by a round of perfect SE and logical readout. Full post-selection (discarding the shots that trigger any detector) is used to obtain the logical error rate and discard rate for each ungrown case. Data corresponding to a full MSC protocol are labeled by `End-to-end'. The end-to-end results are selected from the results in Fig.~\ref{fig: MSC results}. The distance of the code at the end of an MSC protocol is denoted as the final distance. Footprint of an MSC protocol is the maximum number of active qubits throughout the protocol. Our cultivation process alone
 is less efficient than the one in the original MSC protocol. However, thanks to our more efficient ES process, our end-to-end efficiency is higher than the original MSC protocol. }
    \label{tab: MSC results comparison}
\end{table*}

We will elaborate on our MSC protocols and their decoding in the rest of the paper, which is structured as follows. In Sec.~\ref{sec: surface codes on rp2}, we describe the design of $\mathbb{RP}^2$ codes and their SE circuits. In Sec.~\ref{sec: to and fro morphing}, we first present a general framework of transforming a CSS code with fold-duality (a special case of ZX-duality~\cite{breuckmann_fold-transversal_2024}) to a self-dual CSS code. As $\mathbb{RP}^2$ codes have fold-duality, we then define two self-dual codes, named SRP-3 and SRP-5 respectively, and describe how to morph an $\mathbb{RP}^2$-3 ($\mathbb{RP}^2$-5) code to an SRP-3 (SRP-5) code and back. (We will also use SRP to refer to either SRP-3 or SRP-5.) Logical double-checking on an SRP code is described in Sec.~\ref{sec: logical double checking}. Circuits for code growths including injection and expansion are illustrated in Sec.~\ref{sec: injection and growth}. We describe how \textit{two-dimensional soft outputs} are obtained during decoding in Sec.~\ref{sec: final stage post selection}. Finally, we discuss our results in Sec.~\ref{sec: discussion and outlook}.

\section{Surface codes on \texorpdfstring{$\mathbb{RP}^2$}{RP2}}\label{sec: surface codes on rp2}

In this section, we define the $\mathbb{RP}^2$ (surface) code as a variant family of  the planar (rotated) surface code by placing it on $\mathbb{RP}^2$. We note that there is another work~\cite{kobayashi_cross-cap_2024} by Kobayashi et al. that also considered various ways to put the surface code on $\mathbb{RP}^2$. Our code construction is similar to but different from theirs. As per our construction, an $\mathbb{RP}^2$ code with patch parameter $d$ (an odd positive integer), named $\mathbb{RP}^2$-$d$, has code parameters $[[d^2,1,d]]$, just as a rotated surface code with the same patch parameter. For example, stabilizers of the $\mathbb{RP}^2$-$3$ code are illustrated in Fig.~\ref{fig: surface codes on rp2}(\textbf{a}), where the stabilizers in the bulk are the same as the ones in a distance 3 rotated surface code.  More interestingly, each antipodal pair of red (blue) tiles along the boundary in Fig.~\ref{fig: surface codes on rp2}(\textbf{a}) represents a single $\mathrm{X}$ ($\mathrm{Z}$) stabilizer. As shown in Fig.~\ref{fig: surface codes on rp2}(\textbf{b}-\textbf{c}), the logical $\mathrm{X}$ ($\mathrm{Z}$) operator can be hosted on any non-contractable loop on $\mathbb{RP}^2$. More generally, the $\mathbb{RP}^2$-$d$ code can be similarly defined with $d$ being the number of qubits on a side of the square lattice (of data qubits) as in Fig.~\ref{fig: surface codes on rp2}(\textbf{a}).   

For $\mathbb{RP}^2$ codes, we can design SE circuits with depth 5 (only counting the layers of CNOT gates per SE round). As we are particularly interested in the $\mathbb{RP}^2$-$3$ code and the $\mathbb{RP}^2$-$5$ code in the family, their SE circuits are illustrated in Fig~\ref{fig: surface codes on rp2}(\textbf{d}-\textbf{e}). It is worth noticing that the orientation of `zigzag' gate orderings for $\mathrm{X}$ stabilizer measurements in the bulk is aligned with the logical $\mathrm{X}$ operator in Fig.~\ref{fig: surface codes on rp2}(\textbf{b}). While such alignment would halve the circuit-level distance for a rotated surface code (which has open boundaries), in our case, we verified our SE circuit for an $\mathbb{RP}^2$-$d$ code (which does not have a boundary) achieves full circuit-level distance $d$ for all $d\in\{3,5,7,9,11\}$ using \texttt{stim.Circuit.shortest\_graphlike\_error}~\cite{gidney2021stim}.

\begin{figure}[h]
    \centering
    \includegraphics[width=\columnwidth]{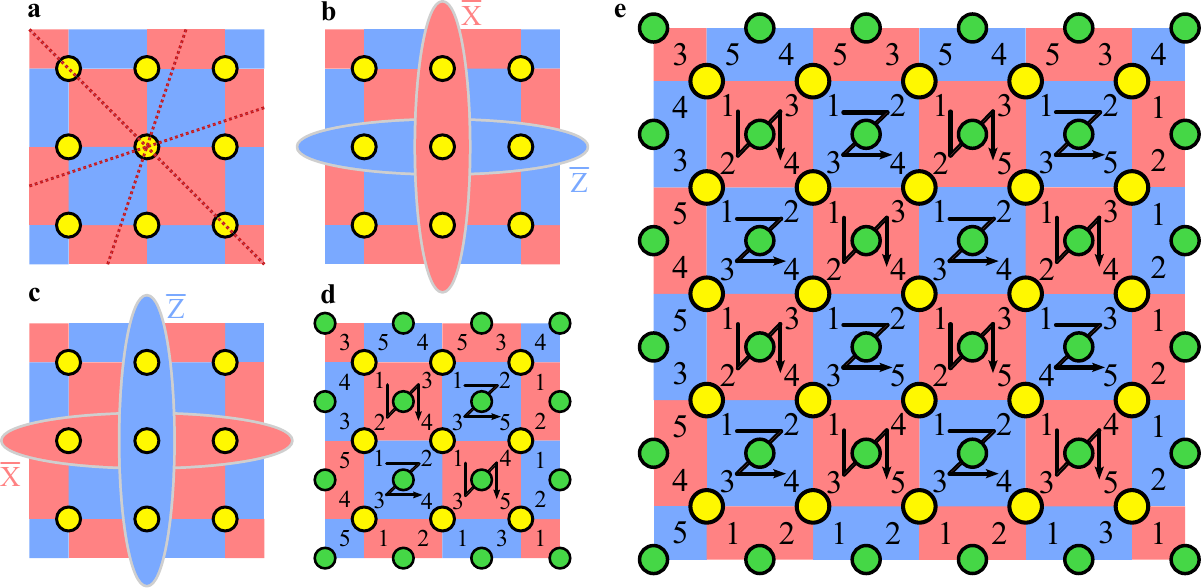}
    \caption{The $\mathbb{RP}^2$ code. (\textbf{a}) Stabilizers of the $\mathbb{RP}^2$-3 code with code parameters $[[9,1,3]]$. Data qubits are marked by yellow circles. Red (blue) tiles indicate Pauli $\mathrm{X}$ ($\mathrm{Z}$) stabilizers. Each pair of antipodal red tiles (connected by red dotted lines) or blue tiles along the boundary are combined together to form a single stabilizer. (\textbf{b}-\textbf{c}) Different representations of the logical $\mathrm{X}$ and $\mathrm{Z}$ operators. (\textbf{d}-\textbf{e}) $\mathrm{CNOT}$ gate orderings of the SE circuits for the $\mathbb{RP}^2$-3 code and the $\mathbb{RP}^2$-5 code respectively. Green circles represent ancilla qubits used during SE. On the boundary, each antipodal pair of ancilla qubits are identified as a single ancilla qubit.    }
    \label{fig: surface codes on rp2}
\end{figure}

\section{To-and-Fro morphing between an \texorpdfstring{$\mathbb{RP}^2$}{RP2}-3 (\texorpdfstring{$\mathbb{RP}^2$}{RP2}-5) and an SRP-3 (SRP-5) code}\label{sec: to and fro morphing}

In the original MSC protocol~\cite{gidney_magic_2024}, the logical $\mathrm{H}_{\mathrm{XY}}$ operator of a $(6,6,6)$-color code is double-checked to detect logical errors. Since $\mathrm{T}|+\rangle$ is the $+1$ eigenstate of $\mathrm{H}_{\mathrm{XY}}$, a $-1$ measurement result of the logical $\mathrm{H}_{\mathrm{XY}}$ signals the presence of errors. On a $(6,6,6)$-color code, the double-checking of logical $\mathrm{H}_{\mathrm{XY}}$ can be done efficiently as it admits a `single-qubit' transversal implementation, meaning that there is a tensor product of single-qubit gates on the data qubits implementing the logical $\mathrm{H}_{\mathrm{XY}}$. For our MSC protocol, we would also like to perform the double-checking of logical $\mathrm{H}_{\mathrm{XY}}$ on our logical qubit. To achieve this, we can first morph the $\mathbb{RP}^2$-3 ($\mathbb{RP}^2$-5) code that holds our logical qubit to an self-dual CSS code (Def.~\ref{def: css codes},~\ref{def: hadamard conjugation},~\ref{def: self dual css codes}) with exactly one logical qubit as such code also admits a `single-qubit' transversal logical $\mathrm{H}_{\mathrm{XY}}$ gate. We can then perform the logical double-checking~\cite{gidney_magic_2024} on the self-dual CSS code. Finally, we can get our logical qubit back on the original $\mathbb{RP}^2$-3 (or $\mathbb{RP}^2$-5) code by morphing back from the self-dual CSS code.  

In this section, we will construct two self-dual CSS codes, the SRP-3 code and the SRP-5 code, as well as shallow circuits that can quickly morph an $\mathbb{RP}^{2}$-3 ($\mathbb{RP}^2$-5) code to an SRP-3 (SRP-5) code and back. Notably, these constructions crucially rely on the fold-duality (Def.~\ref{def: fold duality}) of the $\mathbb{RP}^2$ code. We first present a more general framework for constructing a self-dual CSS code from a fold-dual one in Sec~\ref{subsec: transforming fold-dual codes into self-dual codes via concatenation}. We return to discuss the SRP-3 and the SRP-5 code and the morphing circuit construction in Sec~\ref{subsec: SRP code construction}.

\begin{definition}[CSS codes]\label{def: css codes}
    A stabilizer code $\mathcal{C}$ is a Calderbank-Shor-Steane (CSS) code if $\mathcal{C}$ has a stabilizer generating set in which each element is either a Pauli $\mathrm{X}$ string operator or a Pauli $\mathrm{Z}$ string operator. In the following, we denote the number of data qubits of $\mathcal{C}$ by $|\mathcal{C}|$. We also label the data qubits by $1, 2, \cdots, |\mathcal{C}|$ and use qubits synonymously with their labels. The set of all data qubits is denoted by $[|\mathcal{C}|]:=\{1,2,\cdots,|\mathcal{C}|\}$.  
\end{definition}
\begin{definition}[Hadamard conjugated Pauli string opeartors]\label{def: hadamard conjugation}
For a set $\mathcal{Q}$ of qubits, define the transversal $\mathrm{H}$ gate on $\mathcal{Q}$ as \[
\mathrm{H}_{\mathcal{Q}} := \prod_{i\in\mathcal{Q}}\mathrm{H}_{i},
\] where $\mathrm{H}_i$ is the single-qubit Hadamard gate on qubit $i$. Denote the conjugation of $\mathrm{P}$ by the transversal $\mathrm{H}$ gate on the support of $\mathrm{P}$ as $\mathsf{h}(\mathrm{P}):=\mathrm{H}_{\mathrm{supp}(\mathrm{P})}\mathrm{P}\mathrm{H}_{\mathrm{supp}(\mathrm{P})}^{\dagger}$, where $\mathrm{supp}(\mathrm{P})$ (the support of $\mathrm{P}$) is the set of qubits on each of which $\mathrm{P}$ has a non-identity element. 
\end{definition}
\begin{definition}[Self-dual CSS codes (with even distance)] \label{def: self dual css codes}
    A CSS code $\mathcal{C}$ is said to be self-dual if for each $\mathrm{X}$ ($\mathrm{Z}$) stabilizer $\mathrm{S}$, $\mathsf{h}(\mathrm{S})$ is a $\mathrm{Z}$ ($\mathrm{X}$) stabilizer. A self-dual CSS code $\mathcal{C}$ is said to have \textit{even distance} if $\mathcal{C}$ encodes at least one logical qubit and any X or Z logical operator of $\mathcal{C}$ has even weight. (This definition is natural in the sense that any X (Z) logical operator $\mathrm{L}$ of $\mathcal{C}$ deformed by any X (Z) stabilizer has the same weight parity as $\mathrm{L}$, as a result of the self-duality.) 
\end{definition}
\begin{definition}[Fold-duality (a special case of ZX-dualities defined in~\cite{breuckmann_fold-transversal_2024})] \label{def: fold duality}
     Consider a CSS code $\mathcal{C}$. A \textit{Folding} operation $\sigma$ on $\mathcal{C}$ is defined as an order-2 permutation of code qubits. More precisely, $\sigma$ is defined as a member of the symmetric group $\mathrm{Sym}([|\mathcal{C}|])$ with $\sigma^2=1$. Define the \textit{fold line} of $\sigma$ as $inv(\sigma):=\{i=\sigma(i)|i\in [|\mathcal{C}|]\}$, the set of qubits invariant under $\sigma$.  The unitary operator $\hat{\sigma}$ that implements the \textit{folding} $\sigma$ is simply the product of SWAP gates on each pair of qubits that form a size 2 orbit under $\sigma$:
    \[
   \hat{\sigma}:= \prod_{\substack{(i,\sigma(i)> i), \\i\in[|\mathcal{C}|]}}  \mathrm{SWAP}_{i,\sigma(i)}.
    \]
For a Pauli string operator $\mathrm{P}$, denote its conjugation by $\hat{\sigma}$ as $\hat{\sigma}(\mathrm{P}):=\hat{\sigma}\mathrm{P}\hat{\sigma}^{\dagger}$. 
$\mathcal{C}$ is said to be fold-dual (or to have fold-duality) if there exists a \textit{folding} $\sigma$ on $\mathcal{C}$ that satisfies the following condition: 
\begin{itemize}
\item for each $\mathrm{X}$ ($\mathrm{Z}$) stabilizer $\mathrm{S}$ of $\mathcal{C}$, $\mathsf{h}(\hat{\sigma}(\mathrm{S}))$ is a $\mathrm{Z}$ ($\mathrm{X}$) stabilizer of $\mathcal{C}$. 
\end{itemize}
Conversely, given a fold-dual code $\mathcal{C}$, if a \textit{folding} $\sigma$ on $\mathcal{C}$ satisfies the above condition, we say $\sigma$ is an \textit{associated folding} of $\mathcal{C}$. (Note that $\mathcal{C}$ may have more than one \textit{associated folding}.) 
 \end{definition}

\subsection{Turning fold-dual CSS codes into self-dual ones via concatenation}\label{subsec: transforming fold-dual codes into self-dual codes via concatenation}
Inspired by the previous observation that color codes can be constructed from surface codes~\cite{bombin_universal_2012,kubica_unfolding_2015,criger_concatenated_toric_code_2016} and in particular by the code concatenation perspective in~\cite{criger_concatenated_toric_code_2016}, we will first present a more general recipe for transforming a fold-dual CSS code into a self-dual CSS code by concatenating the fold-dual code with a series of (small) even-distance self-dual CSS codes (Def.~\ref{def: self dual css codes}). Then, to illustrate the recipe more concretely, we return to the examples in~\cite{criger_concatenated_toric_code_2016} where color codes can be constructed by concatenating fold-dual surface codes with a series of $[[6,4,2]]$ and/or $[[4,2,2]]$ codes.

\begin{figure}[h]
    \centering
    \includegraphics[width=\columnwidth]{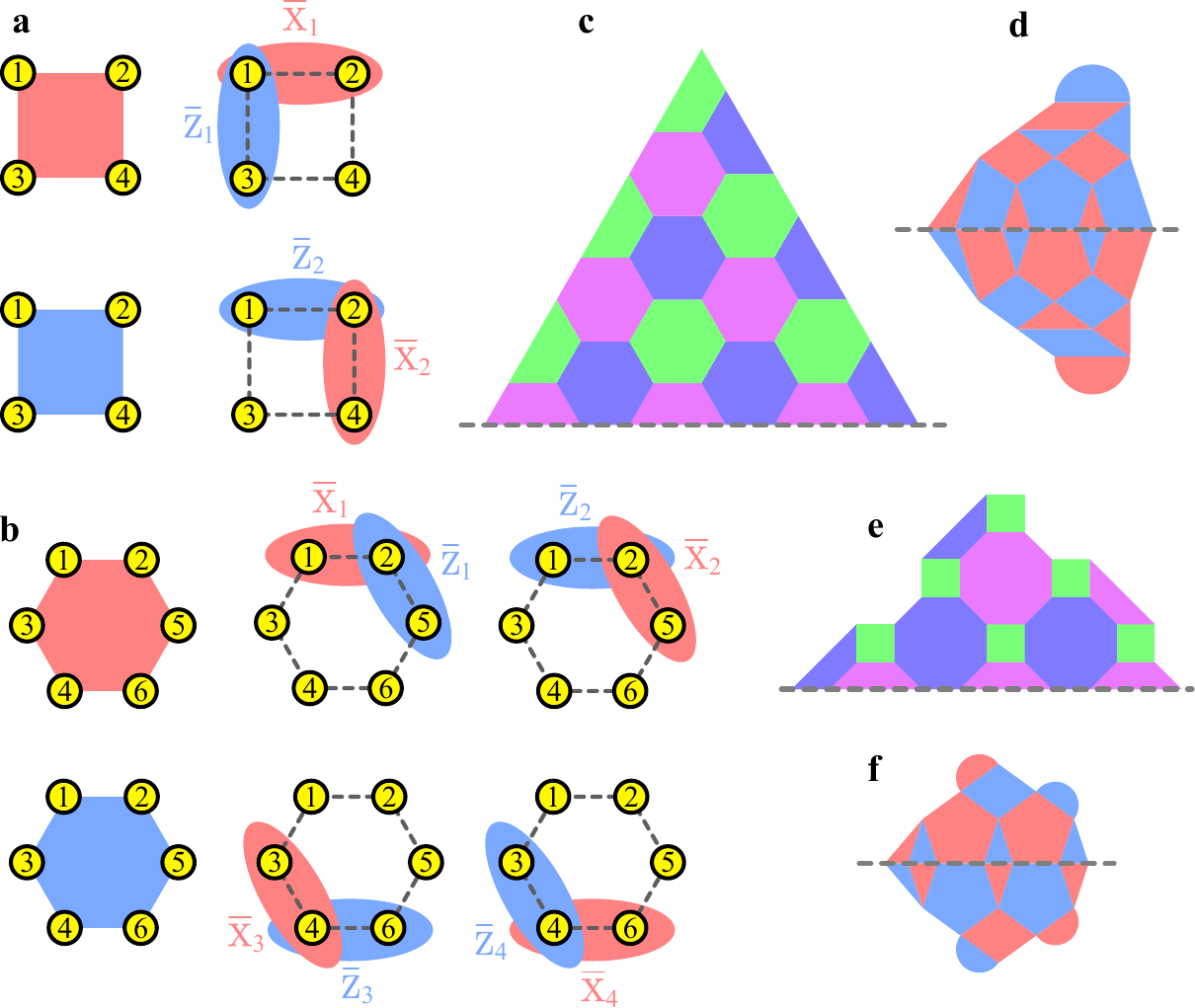}
    \caption{Small even-distance self-dual CSS codes and constructing color codes from surface codes. (\textbf{a}) Stabilizers of the $[[4,2,2]]$ code and its logical operators. Red (blue) indicates Pauli $\mathrm{X}$ ($\mathrm{Z}$) operators.  (\textbf{b}) Stabilizers of the $[[6,4,2]]$ code and its logical operators. (\textbf{c}-\textbf{d}) A distance 7 $(6,6,6)$-color code (in (\textbf{c})) can be constructed from a plane-shaped distance 4 surface code (in (\textbf{d})). (\textbf{e}-\textbf{f}) A distance 7 $(4,8,8)$-color code (in (\textbf{e})) can be constructed from a frog-shaped distance 4 surface code (in (\textbf{f})). Faces of color codes are colored by lavender, magenta and green. The X and Z stabilizer on each green face correspond to the stabilizers of either a $[[4,2,2]]$ code in (\textbf{a}) or a $[[6,4,2]]$ code in (\textbf{b}). Gray dashed lines in (\textbf{c-d}) (or in (\textbf{e-f})) depict the \textit{fold line} of an \textit{associated folding} of the plane-shaped (or frog-shaped) surface code. }
    \label{fig: small even distance self-dual codes and unfolding the color codes}
\end{figure}

We will use the following property of even distance self-dual CSS codes. 
\begin{lemma}
    [Logical effect of transversal H gates on an even distance self-dual CSS code] \label{lemma: properties of even distance self-dual CSS codes}
    Consider an even distance self-dual CSS code $\mathcal{C}$ with code parameters $[[n:=|\mathcal{C}|,k,d]]$. Then $n$, $k$ and $d$ must be even numbers. Suppose $\mathcal{C}$ have $2k$ logical qubits, then they can be chosen and arranged such that transversally applying H gates on all data qubits of $\mathcal{C}$  implements the following logical unitary:
    \begin{equation}\label{eq: logical unitary by transversal H}
         \prod_{i=1}^{k}\left( \overline{\mathrm{H}}_{2i-1} \otimes  \overline{\mathrm{H}}_{2i} \circ \overline{\mathrm{SWAP}}_{2i-1,2i} \right)
    \end{equation}
\end{lemma}
\begin{proof}
    The idea of this proof is to recursively construct pairs of logical qubits such that the logical effect of the transversal H gate $\mathrm{H}_{[|\mathcal{C}|]}$ on the $i$-th pair (containing qubit $2i-1$ and qubit $2i$) is
    \begin{equation}\label{eq: logical effect on qubit pairs}
         \overline{\mathrm{H}}_{2i-1}\otimes \overline{\mathrm{H}}_{2i}\circ \overline{\mathrm{SWAP}}_{2i-1,2i}. 
    \end{equation}
   In this way, we can prove $k$ is even and also the logical effect of $\mathrm{H}_{[|\mathcal{C}|]}$ is given by Eq.~\ref{eq: logical unitary by transversal H}. Moreover, as the number of independent X stabilizers is the same as the number of independent Z stabilizers, $n$ has the same parity as $k$ and thus is also even. 

    First, we start with an arbitrarily chosen logical qubit (labeled by $1$) with logical $\mathrm{X}$ ($\mathrm{Z}$) operator $\overline{\mathrm{X}}_1$ ($\overline{\mathrm{Z}}_1$). As $|\overline{\mathrm{X}}_1|$ is even, $\mathsf{h}\left(\overline{\mathrm{X}}_1\right)$ commutes with $\overline{\mathrm{X}}_1$ and thus is a logical $\mathrm{Z}$ operator inequivalent to $\overline{\mathrm{Z}}_1$. We define $\overline{\mathrm{Z}}_2$ as $\mathsf{h}\left(\overline{\mathrm{X}}_1\right)$ and similarly $\overline{\mathrm{X}}_2$ as $\mathsf{h}\left(\overline{\mathrm{Z}}_1\right)$. We see that $\{\overline{\mathrm{X}}_1,\overline{\mathrm{Z}}_1,\overline{\mathrm{X}}_2,\overline{\mathrm{Z}}_2\}$ specifies a pair of logical qubits which satisfies the requirement in Eq.~\ref{eq: logical effect on qubit pairs}. 

    Now, suppose we have constructed $n$ pairs of logical qubits such that within each pair $j$ (containing qubit $2j-1$ and qubit $2j$), 
    \[
    \mathsf{h}(\overline{\mathrm{X}}_{2j-1}) = \overline{\mathrm{Z}}_{2j}, \quad \text{and}\quad \mathsf{h}(\overline{\mathrm{Z}}_{2j-1}) = \overline{\mathrm{X}}_{2j}.
    \]
    If $\mathcal{C}$ encodes only $2n$ logical qubits, then the proof is complete. Otherwise, we can find another logical qubit $2n+1$ such that its logical operators $\overline{\mathrm{X}}_{2n+1}$ and $\overline{\mathrm{Z}}_{2n+1}$ commute with $\overline{\mathrm{Z}}_{l}$ and $\overline{\mathrm{X}}_{l}$ respectively for $\forall\  l\in\{1,2,\cdots,2n\}$. Moreover, $\mathsf{h}\left(\overline{\mathrm{X}}_{2n+1}\right)$ ($\mathsf{h}\left(\overline{\mathrm{Z}}_{2n+1}\right)$) is also a logical operator that commutes with $\overline{\mathrm{X}}_{m}$ ($\overline{\mathrm{Z}}_{m}$) for $\forall m\in\{1,2,\cdots,2n+1\}$. In this way, a new qubit pair specified by $\{\overline{\mathrm{X}}_{2n+1},\overline{\mathrm{Z}}_{2n+1},\overline{\mathrm{X}}_{2n+2}:=\mathsf{h}\left(\overline{\mathrm{Z}}_{2n+1}\right),\overline{\mathrm{Z}}_{2n+2}:=\mathsf{h}\left(\overline{\mathrm{X}}_{2n+1}\right)\}$ is constructed and satisfies the requirement in Eq.~\ref{eq: logical effect on qubit pairs}. 
    
\end{proof}

Based on Lemma.~\ref{lemma: properties of even distance self-dual CSS codes}, we tacitly assume in the rest of the paper that for any even distance self-dual CSS code $\mathcal{C}$ with $2k$ logical qubits, its logical qubits are already arranged such that the transversal $\mathrm{H}$ gate on $\mathcal{C}$, $\mathrm{H}_{[|\mathcal{C}|]}$, implements the logical unitary in Eq.~\ref{eq: logical unitary by transversal H}. Also, for each $i\in\{1,2,\cdots,k\}$, the logical qubits $2i-1$ and $2i$ are grouped into a logical qubit pair. 

\begin{proposition}
    [Transforming a fold-dual CSS code into a self-dual CSS code via concatenation] 
    \label{prop: construct self-dual css from fold-dual css}
    Given a fold-dual CSS $[[n,k,d]]$ code $\mathcal{C}$ with associated \textit{folding} $\sigma$, and a series of even weight self-dual CSS codes $\mathcal{C}_a$ for  $a=1,2,\cdots,l$  with code parameters $[[n_a,k_a,d_a]]$ ($n_a$, $k_a$ and $d_a$ are all positive even numbers) for each $a$, then as long as 
    \[n-|inv(\sigma)| = \sum_{a}k_a,\]
    we can construct a self-dual CSS code $\tilde{\mathcal{C}}$ by encoding each data qubit pair $\{i,\sigma(i)>i\}$ in $\mathcal{C}$ as a logical qubit pair in $\mathcal{C}_{a}$ for some $a\in\{1,2,\cdots,l\}$. 
\end{proposition}
\begin{proof}
    Define encoding map $\mathsf{Enc}$ that maps Pauli string operators on (data qubits of) $\mathcal{C}$ to Pauli string operators on $\tilde{\mathcal{C}}$ such that for each Pauli string operator $\mathrm{P}$ on $\mathcal{C}$, $\mathsf{Enc}(\mathrm{P})$ is the encoded Pauli string operator on $\tilde{\mathcal{C}}$. Denote the X stabilizer generating set of $\mathcal{C}$ and $\mathcal{C}_a$ ($a\in\{1,2,\cdots,l\}$) as $\mathcal{S}_{\mathrm{X}}$ and $\mathcal{S}_{\mathrm{X},a}$ respectively. The generating set of X stabilizers for $\tilde{\mathcal{C}}$ can be chosen as $\tilde{\mathcal{S}}_{\mathrm{X}}=(\sqcup_{a}\mathcal{S}_{\mathrm{X},a})\sqcup \mathsf{Enc}(\mathcal{S}_{\mathrm{X}})$, with $\mathsf{Enc}(\mathcal{S}_{\mathrm{X}}):=\{\mathsf{Enc}(\mathrm{S})|\mathrm{S}\in\mathcal{S}_{\mathrm{X}}\}$.  

Consider an X stabilizer $\tilde{\mathrm{S}}\in\tilde{\mathcal{S}}_{\mathrm{X}}$. If $\tilde{\mathrm{S}}\in\sqcup_{a}\mathcal{S}_{\mathrm{X},a}$, then $\mathrm{h}(\tilde{\mathrm{S}})$ is a $\mathrm{Z}$ stabilizer of $\tilde{\mathcal{C}}$ due to self-duality of the codes $\mathcal{C}_{a}$. On the other hand, if $\tilde{\mathrm{S}}\in\mathsf{Enc}(\mathcal{S}_{\mathrm{X}})$, there exists $\mathrm{S}\in\mathcal{S}_{\mathrm{X}}$ such that $\tilde{\mathrm{S}}=\mathsf{Enc}(\mathrm{S})$. For a qubit $c\in inv(\sigma)$ on the \textit{fold line} of $\sigma$, $\mathsf{h}(\mathsf{Enc}(\mathrm{X}_{c}))=\mathsf{Enc}(\mathrm{Z}_{c})$. For a qubit $t\in[|\mathcal{C}|]\backslash inv(\sigma)$ not on the \textit{fold line}, $\mathsf{h}(\mathsf{Enc}(\mathrm{X}_{t}))=\mathsf{Enc}(\mathrm{Z}_{\sigma(t)})$. Thus
\begin{align}
\mathsf{h}(\tilde{\mathrm{S}}) &= \prod_{c\in\mathrm{supp}(\mathrm{S})\cap inv(\sigma)} \mathsf{Enc}(\mathrm{Z}_{c}) \cdot \prod_{t\in\mathrm{supp}(\mathrm{S})\backslash inv(\sigma)} \mathsf{Enc}(\mathrm{Z}_{\sigma(t)}) \nonumber \\
&= \mathsf{Enc}(\mathsf{h}(\hat{\sigma}(\mathrm{S}))) 
\end{align}
Due to the fold-duality of $\mathcal{C}$, $\mathsf{h}(\hat{\sigma}(\mathrm{S}))$ is a $\mathrm{Z}$ stabilizer of $\mathcal{C}$. Thus $\mathsf{h}(\tilde{\mathrm{S}})$ is a $\mathrm{Z}$ stabilizer of $\tilde{\mathcal{C}}$. Similarly, we can also prove that for each $\mathrm{Z}$ stabilizer $\tilde{\mathrm{S}}$ of $\tilde{\mathcal{C}}$, $\mathrm{h}(\tilde{\mathrm{S}})$ is an $\mathrm{X}$ stabilizer of $\tilde{\mathcal{C}}$. Hence $\tilde{\mathcal{C}}$ is self-dual.  
\end{proof}

\begin{figure}
    \centering
    \includegraphics[width=\columnwidth]{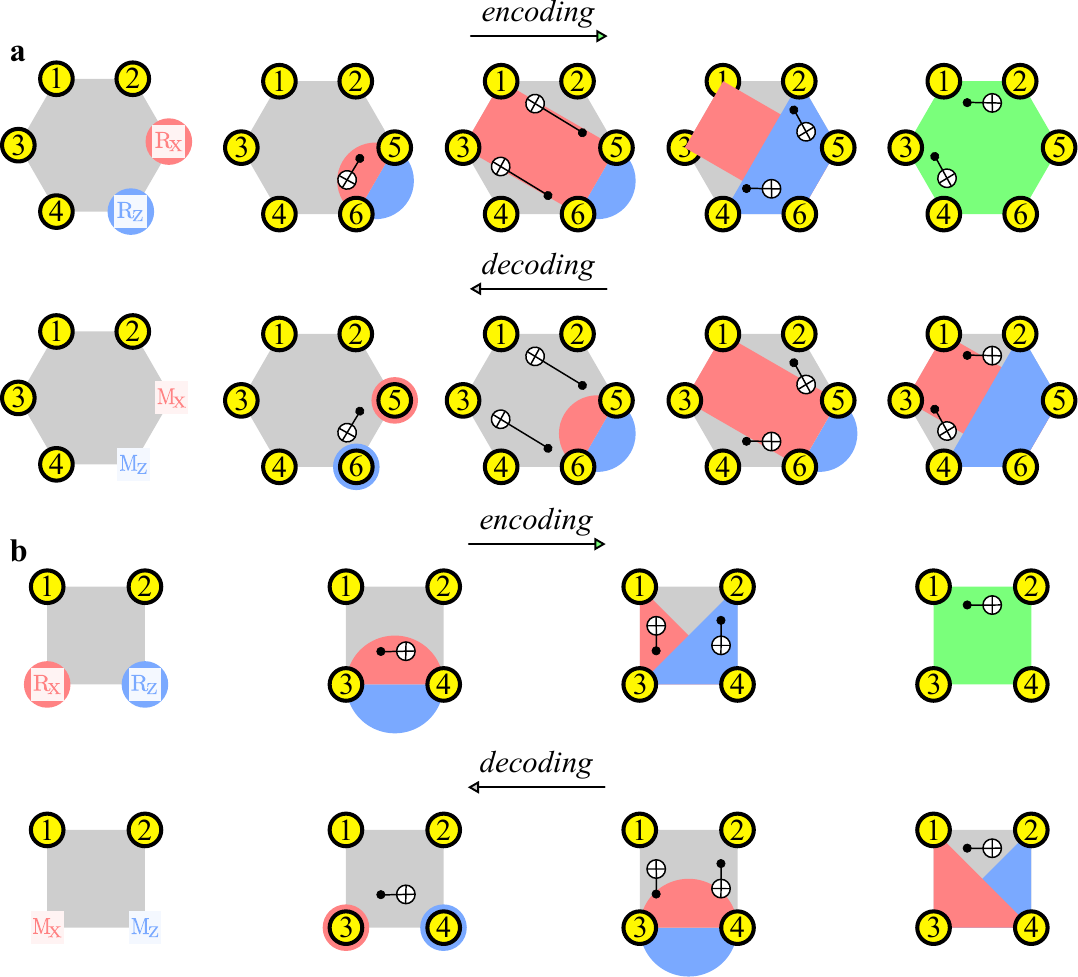}
    \caption{Building blocks that enable fast morphing between an $\mathbb{RP}^2$-3 ($\mathbb{RP}^2$-5) code and an SRP-3 (SRP-5) code. (\textbf{a}) Encoding and decoding circuits for the $[[6,4,2]]$ code. (\textbf{b}) Encoding and decoding circuits for the $[[4,2,2]]$ code. Assuming the decoding circuit is 
 run right after the encoding circuit, we track the detecting regions~\cite{mcewen_relaxing_2023} generated by reset locations throughout the encoding and decoding circuit in (\textbf{a}) and (\textbf{b}). Each green face indicates a Pauli $\mathrm{X}$ string operator \textit{and} a Pauli $\mathrm{Z}$ string operator.}
    \label{fig: encoding and decoding circuits of small self dual codes}
\end{figure}

\begin{figure}[ht]
    \centering
    \includegraphics[width=\columnwidth]{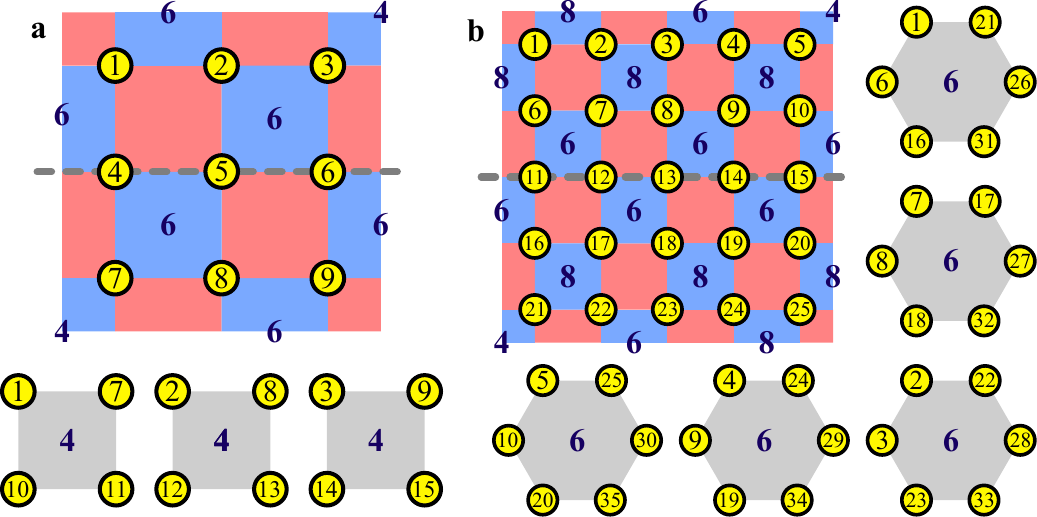}
    \caption{Operational definition of the SRP-3 and SRP-5 code. Each $\mathbb{RP}^2$ code has an \textit{associated folding} with its \textit{fold line} illustrated by a grey dashed line here. (\textbf{a}) The $[[15,1,3]]$ SRP-3 code is obtained by running the encoding circuit in Fig.~\ref{fig: encoding and decoding circuits of small self dual codes}(\textbf{b}) on each gray square. Here qubits 1-9 are the data qubits of an $\mathbb{RP}^2$-3 code; qubits 10-15 are ancilla qubits introduced at the beginning of the encoding circuits. (\textbf{b}) The $[[35,1,5]]$ SRP-5 code is obtained by running the encoding circuit in Fig.~\ref{fig: encoding and decoding circuits of small self dual codes}(\textbf{a}) on each gray hexagon. Here qubits 1-25 are the data qubits of an $\mathbb{RP}^2$-5 code; qubits 26-35 are ancilla qubits introduced at the beginning of the encoding circuits. The weight of each $\mathrm{Z}$ stabilizer of the SRP-3 (SRP-5) code is presented in the following way. For each $\mathrm{Z}$ stabilizer of the $\mathbb{RP}^2$-3 ($\mathbb{RP}^2$-5) code, the weight of the corresponding encoded $\mathrm{Z}$ stabilizer of the SRP-3 (SRP-5) code is illustrated at the center of the original (un-encoded) $\mathrm{Z}$ stabilizer. The weight of the $\mathrm{Z}$ stabilizer of each $[[4,2,2]]$ and $[[6,4,2]]$ code is illustrated at the center of each gray square and hexagon respectively.}
    \label{fig: srp3 and srp5 codes}
\end{figure}

From Prop.~\ref{prop: construct self-dual css from fold-dual css}, we see that the constructed self-dual code $\tilde{\mathcal{C}}$ is composed of data qubits from $\mathcal{C}$ that remain invariant under \textit{folding} $\sigma$ and data qubits from the small self-dual codes. The (small) even distance self-dual CSS codes used for concatenation in~\cite{criger_concatenated_toric_code_2016} are the $[[4,2,2]]$ code and the $[[6,4,2]]$ code (Fig.~\ref{fig: small even distance self-dual codes and unfolding the color codes}(\textbf{a}-\textbf{b})). A $(6,6,6)$-color code (Fig.~\ref{fig: small even distance self-dual codes and unfolding the color codes}(\textbf{c})) with distance $7$ can be constructed from a plane-shaped fold-dual surface code (Fig.~\ref{fig: small even distance self-dual codes and unfolding the color codes}(\textbf{d})) with distance 4 via concatenating the surface code with three copies of the $[[6,4,2]]$ codes and three copies of the $[[4,2,2]]$ codes. Each green face of the color code in Fig.~\ref{fig: small even distance self-dual codes and unfolding the color codes}(\textbf{c}) represents either a $[[6,4,2]]$ code or a $[[4,2,2]]$ code. Similarly, as shown in Fig.~\ref{fig: small even distance self-dual codes and unfolding the color codes}(\textbf{e}-\textbf{f}), a $(4,8,8)$-color code with distance $7$ can be constructed from a frog-shaped fold-dual surface code with distance $4$ by concatenating it with six copies of the $[[4,2,2]]$ codes, also marked by green faces. More generally, both $(6,6,6)$- and $(4,8,8)$-color codes with distance $2t-1$ can be constructed similarly using surface codes with distance $t$. (Notice that any color code that encode only one logical qubit must have an odd distance due to Lemma~\ref{lemma: properties of even distance self-dual CSS codes}.) We note that there are other proposals that also use the technique of concatenation on top of $[[4,2,2]]$ or $[[6,4,2]]$ codes~\cite{knill_2005,yoshida2024concatenatecodessavequbits,goto_2024,berthusen2025adaptivesyndromeextraction}. Our work add to these results by applying this technique in growing magic states.


\subsection{Construction of SRP codes and morphing between them and \texorpdfstring{$\mathbb{RP}^2$}{RP} codes}\label{subsec: SRP code construction}

In this subsection, we will operationally define the SRP-3 (SRP-5) code by specifying how data qubits from an $\mathbb{RP}^2$-3 ($\mathbb{RP}^2$-5) code are encoded into $[[4,2,2]]$ ($[[6,4,2]]$) codes. We start by describing the encoding and decoding circuits for $[[4,2,2]]$ and $[[6,4,2]]$ codes. Morphing between $\mathbb{RP}^2$ codes and SRP codes can be implemented by running these encoding or decoding circuits according to our operational definition of the SRP codes.

As shown in Fig.~\ref{fig: encoding and decoding circuits of small self dual codes}(\textbf{a}-\textbf{b}), we designed encoding circuits for the $[[6,4,2]]$ and the $[[4,2,2]]$ code respectively so that a physical qubit labeled by $i$ ($i\in\{1,2,3,4\}$ for the $[[6,4,2]]$ code and $i\in\{1,2\}$ for the $[[4,2,2]]$ code) is morphed to logical qubit $i$ in Fig.~\ref{fig: small even distance self-dual codes and unfolding the color codes}.  The encoding circuits work by initializing a Bell pair on two ancilla qubits and growing its stabilizers to the stabilizers of either the $[[6,4,2]]$ code or the $[[4,2,2]]$ code. In particular, the structure of the encoding circuit for the $[[6,4,2]]$ code ensures errors within the qubit block $\{1,2,5\}$ or $\{3,4,6\}$ will only propagate within the block. (We note that the encoding circuit for the $[[6,4,2]]$ code used here is very similar to the one in~\cite{goto_2024}.) In our work, we will morph an $\mathbb{RP}^2$ code to an SRP code and then back after a logical double-check on the SRP code. The backwards morphing is implemented by the decoding circuits, which are encoding circuits run in reverse with resets replaced by measurements (Fig.~\ref{fig: encoding and decoding circuits of small self dual codes}). Since the logical double-checking is supposed to leave the stabilizers of the SRP code invariant,  the ancilla qubits measured at the end of the decoding circuits also serve as flag qubits to detect errors during the encoding, logical double-checking and decoding stage.

Notice that each $\mathbb{RP}^2$ code is fold-dual and has an \textit{associated folding} whose \textit{fold line} is the set of qubits on the horizontal line through the center of the $\mathbb{RP}^2$ code patch (Fig~\ref{fig: srp3 and srp5 codes}). Based on this observation, we define the SRP-3 and SRP-5 code operationally in Fig.~\ref{fig: srp3 and srp5 codes}.  The to-be-encoded data qubits of the $\mathbb{RP}^2$-3 code and extra ancilla qubits are arranged on gray faces at the moment immediately  before the encoding step (Fig.~\ref{fig: srp3 and srp5 codes}(\textbf{a})).  The SRP-3 code is defined as the resulting code after applying the encoding circuit on each gray face according to Fig.~\ref{fig: small even distance self-dual codes and unfolding the color codes}(\textbf{b}). The SRP-5 code is defined similarly in Fig.~\ref{fig: srp3 and srp5 codes}(\textbf{b}).

\begin{figure*}
    \centering
    \includegraphics[width=\textwidth]{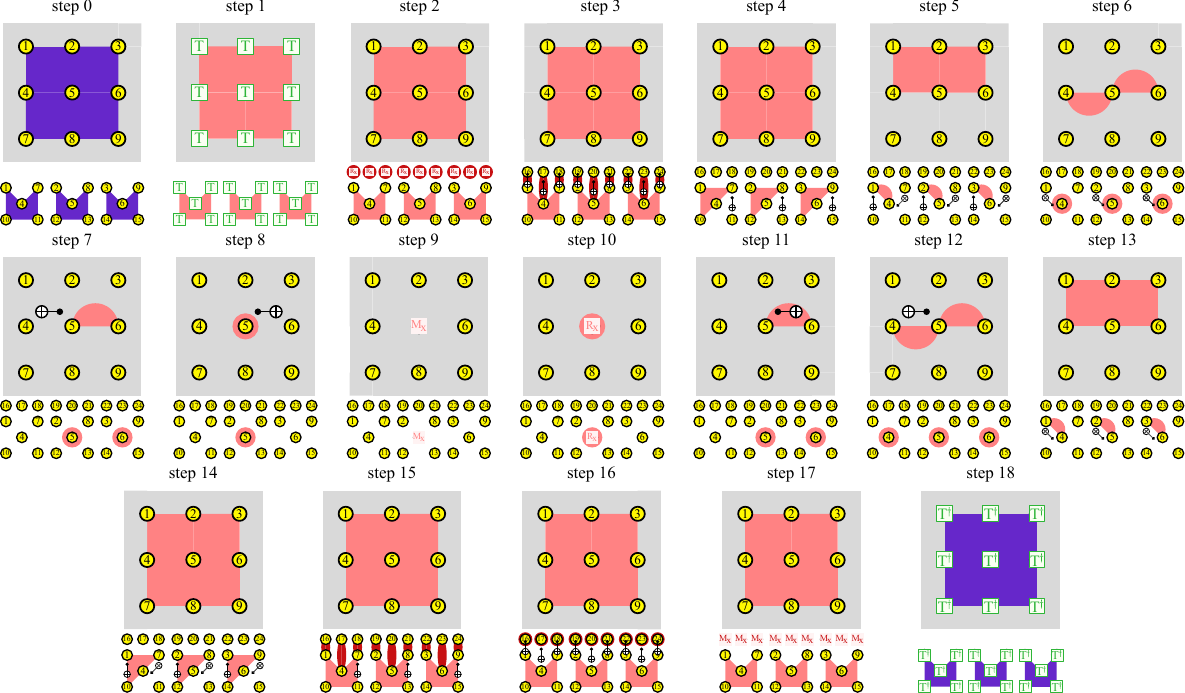}
    \caption{Double-checking circuit for the SRP-3 code. At step 0, purple regions illustrate the operator $\tilde{\mathrm{H}}^{\otimes 15}$ that we need to measure. At step 1, $\tilde{\mathrm{H}}^{\otimes 15}$ is transformed to $\mathrm{X}^{\otimes 15}$ whose evolution is illustrated by (light) red regions and tracked until step 17. Detector regions generated by the flag qubits (initialized at step 2) are colored by dark red and are only shown in step 2, 3, 15 and 16. The purple regions at step 18 indicate the state here should be a $+1$ eigenstate of $\tilde{\mathrm{H}}^{\otimes 15}$ for an ideal double-checking circuit.}
    \label{fig: double checking for SRP-3}
\end{figure*}

\begin{figure*}
    \centering
    \includegraphics[width=\textwidth]{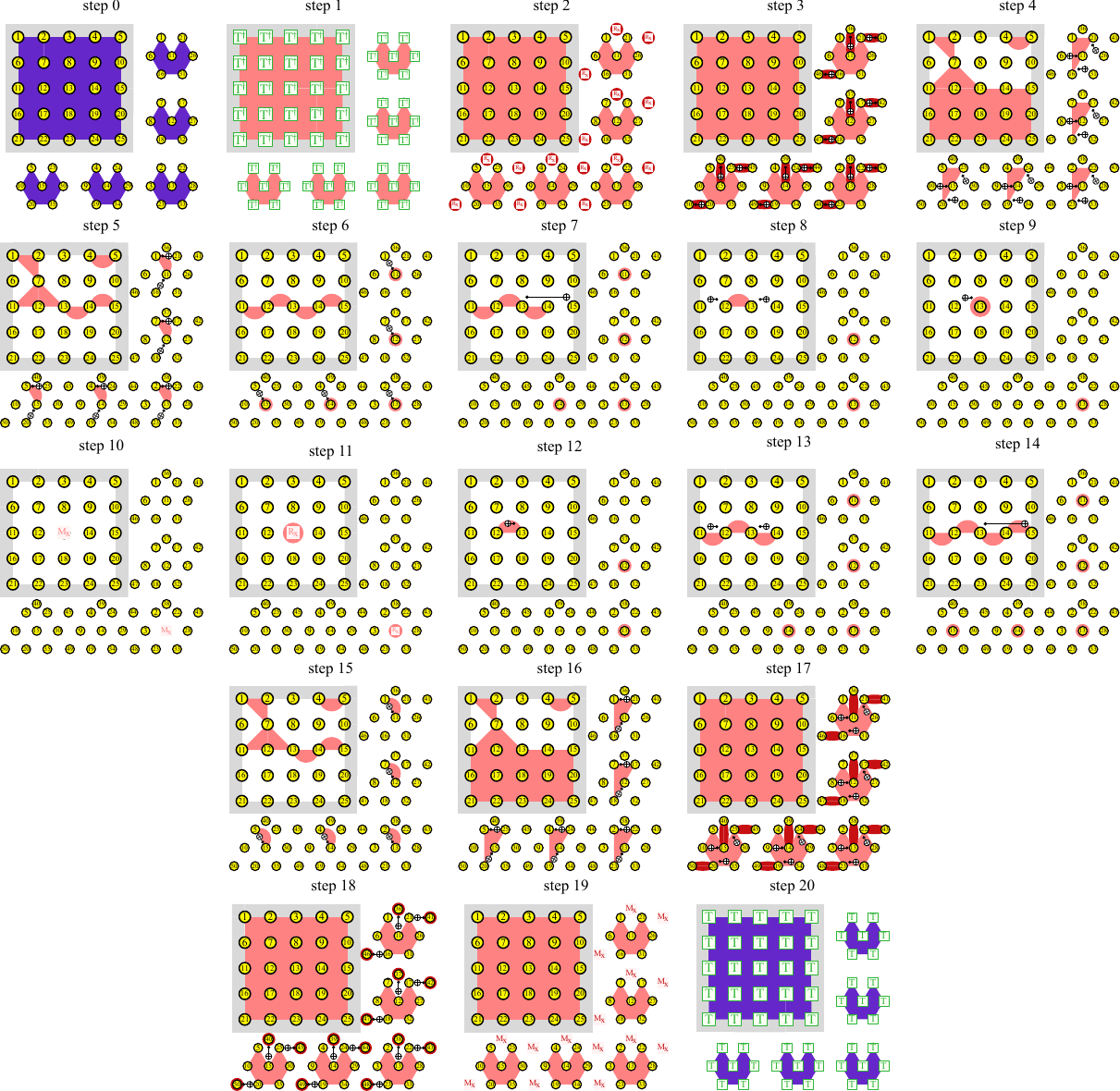}
    \caption{Double-checking circuit for the SRP-5 code. At step 0, purple regions represent the operator $\mathrm{H}_{\mathrm{XY}}^{\otimes 35}$ that we need to measure. At step 1, $\mathrm{H}_{\mathrm{XY}}^{\otimes 35}$ is transformed to $\mathrm{X}^{\otimes 35}$ whose evolution is illustrated by (light) red regions and tracked until step 19. Detector regions generated by the flag qubits (initialized at step 2) are colored by dark red and are only shown in step 2, 3, 17 and 18. The purple regions at step 20 indicate the state here should be a $+1$ eigenstate of $\mathrm{H}^{\otimes 35}$ for an ideal double-checking circuit. One key difference of the circuit here form the double-checking circuit for the SRP-3 code is that a layer of $\mathrm{T}^{\dagger}$ gates (instead of a layler of $\mathrm{T}$ gates for the SRP-3 code (Fig.~\ref{fig: double checking for SRP-3})) is applied at step 1. The reason for such difference is that the operator to be measured is $\mathrm{H}_{\mathrm{XY}}^{\otimes 35}$ (instead of $\tilde{\mathrm{H}}_{\mathrm{XY}}^{\otimes 15}$ for the SRP-3 code). }
    \label{fig: double checking for SRP-5}
\end{figure*}

\section{Logical Double-Checking of SRP-3 and SRP-5 codes}\label{sec: logical double checking}

To detect logical errors, we design the circuits for double-checking the logical $\mathrm{H}_{\mathrm{XY}}$ operator of the SRP-3 and SRP-5 code respectively,  
following~\cite{gidney_magic_2024}. For the SRP-3 code, transversal $\tilde{\mathrm{H}}_{\mathrm{XY}}$ ($\tilde{\mathrm{H}}_{\mathrm{XY}}:=\frac{1}{\sqrt{2}}(\mathrm{X}-\mathrm{Y})$) on all data qubits implements the logical $\mathrm{H}_{\mathrm{XY}}$. On the other hand, for the SRP-5 code, the logical $\mathrm{H}_{\mathrm{XY}}$ is implemented by transversal $\mathrm{H}_{\mathrm{XY}}$ on all data qubits. Also, to make sure such transversal implementation only affects the logical qubit and does not flip the sign of any stabilizer, we require the SRP-3 (SRP-5) code prepared in the $+1$ eigenspace of Z stabilizers with doubly-even weight and in the $-1$ eigenspace of Z stabilizers with singly-even weight~\cite{gidney_magic_2024}. We specify how to achieve this with feed-forward Pauli X correction in Sec.~\ref{sec: injection and growth}. In the following, we describe our double-checking circuits in detail. 

The double-checking circuit for the SRP-3 (SRP-5) code is composed of two components: the logical processing component for double-checking the logical operator and the flag component for detecting errors during the former component. For the SRP-3 code (Fig.~\ref{fig: double checking for SRP-3}), the logical processing component is composed of step 0-1, 4-15 and 18. To measure the $\tilde{\mathrm{H}}_{\mathrm{XY}}^{\otimes 15}$ operator, we first transform it to $\mathrm{X}^{\otimes 15}$ by applying a layer of $\mathrm{T}$ gates on all qubits. A series of CNOT gates are applied to contract $\mathrm{X}^{\otimes 15}$ to a single Pauli $\mathrm{X}$ on qubit 5, which is then measured in the $\mathrm{X}$ basis to read out the logical measurement result. The measurement circuit is then run in reverse to restore the code as well as the logical qubit. As for the flag component, at step 2-3, we insert flag qubits, which are then decoupled from the circuits at step 16 and measured at step 17. Note that the flag component only checks for Z errors and leave X errors unchecked. The latter may trigger detectors in the following SE rounds but do not induce logical error by themselves. This can be readily checked by running an ideal double-checking circuit and then noiselessly measure all stabilizers. We note that there is a lot of freedom in designing the double-checking circuits. For our design, we prioritize conceptual simplicity and explicit parallelizability (at each step shown in Fig.~\ref{fig: double checking for SRP-3}, we have the same operation on each bottom square). We leave a more compact design for future work. The double checking circuit for the SRP-5 code is designed based on the same principle and is shown in Fig.~\ref{fig: double checking for SRP-5}.

\section{Injection and growth}\label{sec: injection and growth}
\begin{figure*}
    \centering
    \includegraphics[width=\textwidth]{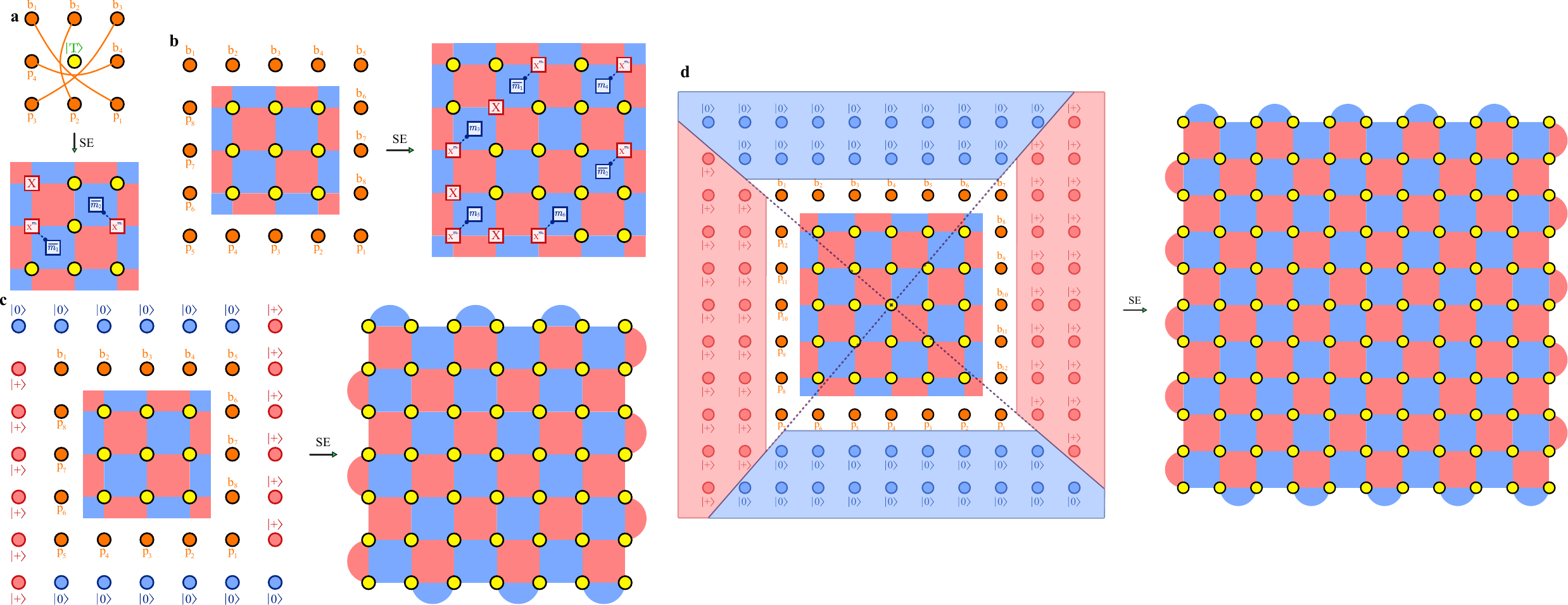}
    \caption{Circuits for code growth. Each qubit pair $\{b_i,p_i\}$ ($i\in\mathbb{N}$) is a Bell pair. Stabilizer measurement result is denoted by $m_a$ or $\overline{m_a}$ (here $m_a$ records the inverse of the actual measurement result), for some $a\in\mathbb{N}$, at the center of the measured stabilizer.   (\textbf{a}) Injection circuit. Qubits from the same Bell pair are connected by an orange curve. After an SE round, a single qubit $\mathrm{T}$ state is injected into an $\mathbb{RP}^2$-3 code. (\textbf{b}) Growing an $\mathbb{RP}^2$-3 code to an $\mathbb{RP}^2$-5 code. In both (\textbf{a}) and (\textbf{b}), Pauli $\mathrm{X}$ corrections are applied at the end of the SE round to fix the signs of the Z stabilizers. A Pauli $\mathrm{X}$ operation is linked by a dashed blue line to a stabilizer measurement result iff it controls this Pauli $\mathrm{X}$ operation. (\textbf{c}) Expanding an $\mathbb{RP}^2$-3 code to a distance 7 rotated surface code. (\textbf{d}) Expanding an $\mathbb{RP}^2$-5 code to a distance 11 rotated surface code. Four blocks of initialized $|0\rangle$ states and $|+\rangle$ states are illustrated by light blue and light red regions respectively.} 
    \label{fig: circuits for patch growth}
\end{figure*}
In this section, we describe three types of circuits for patch growth: an injection circuit (for injecting a single-qubit state into a $\mathbb{RP}^2$-3 code), an $\mathbb{RP}^2$-growth circuit (for growing an $\mathbb{RP}^2$-3 code to an $\mathbb{RP}^2$-5 code) and expansion circuits (for expanding an $\mathbb{RP}^2$ code to a larger rotated surface code). All our growth circuits grow a smaller code at center outwards and make use of Bell pairs.  The injection circuit is implemented by a single round of SE on top of a single qubit $|\mathrm{T}\rangle$ at the center and four surrounding Bell pairs (Fig.~\ref{fig: circuits for patch growth}(\textbf{a})). We fix the signs of the Z stabilizers of the $\mathbb{RP}^2$-3 code by applying Pauli $\mathrm{X}$ corrections based on the Z stabilizer measurement results (Fig.~\ref{fig: circuits for patch growth}(\textbf{a})). Similar to $\mathrm{T}$ state injection, growing the $\mathbb{RP}^2$-3 code to a $\mathbb{RP}^2$-5 code and fixing its Z stabilizer signs can be implemented according to Fig.~\ref{fig: circuits for patch growth}(\textbf{b}). Expanding an $\mathbb{RP}^2$-3 code to a distance 7 rotated surface code requires preparing a layer of Bell pairs surrounding the $\mathbb{RP}^2$ code and also initializing four blocks of $|0\rangle$ or $|+\rangle$ states (Fig.~\ref{fig: circuits for patch growth}(\textbf{c})). Expanding an $\mathbb{RP}^2$-5 code to a distance 11 rotated surface code is also done in a similar fashion (Fig.~\ref{fig: circuits for patch growth}(\textbf{d}). In our MSC-3 protocol, the expansion circuit is run concurrently with the  SRP-3-to-$\mathbb{RP}^2$-3 morphing circuit (Appendix~\ref{appsec: extra details}). Similarly, in our MSC-5 protocol, the $\mathbb{RP}^2$-growth circuit is run concurrently with the SRP-3-to-$\mathbb{RP}^2$-3 morphing circuit and the expansion circuit for the $\mathbb{RP}^2$-5 code is run concurrently with the SRP-5-to-$\mathbb{RP}^2$-5 morphing circuit.   

\section{Final-stage post-selection based on soft outputs}\label{sec: final stage post selection} 
During the expansion stage as well as its first few following SE rounds, the logical qubit is not yet fully protected at the larger surface code distance and is still more susceptible to errors. This necessitates another round of  post-selection at the end of the MSC protocol. Prior works~\cite{gidney_magic_2024,vaknin_magic_2025} on magic state cultivation uses the complementary gap method~\cite{gidney_yoked_2023} to estimate, for each shot, the decoding confidence which is used to determine whether this shot should be discarded. In this work, an alternative method named the soft-output method developed in~\cite{meister_efficient_2024} is used to obtain a measure of decoding confidence. In particular, following~\cite{meister_efficient_2024,gidney_how_to_2024}, we customized the PyMatching package~\cite{higgott_sparse_2025} to compute the soft outputs from internal variables created during the decoding process for each shot.  The soft-output method is more versatile and suitable for our case, which involves  boundary-less $\mathbb{RP}^2$ codes, as this method does not require any specific boundary condition of the code (or the corresponding decoding graph) to work while the complementary gap method requires an open boundary condition possessed by (rotated) surface codes. 

\begin{figure}
    \centering
    \includegraphics[width=\columnwidth]{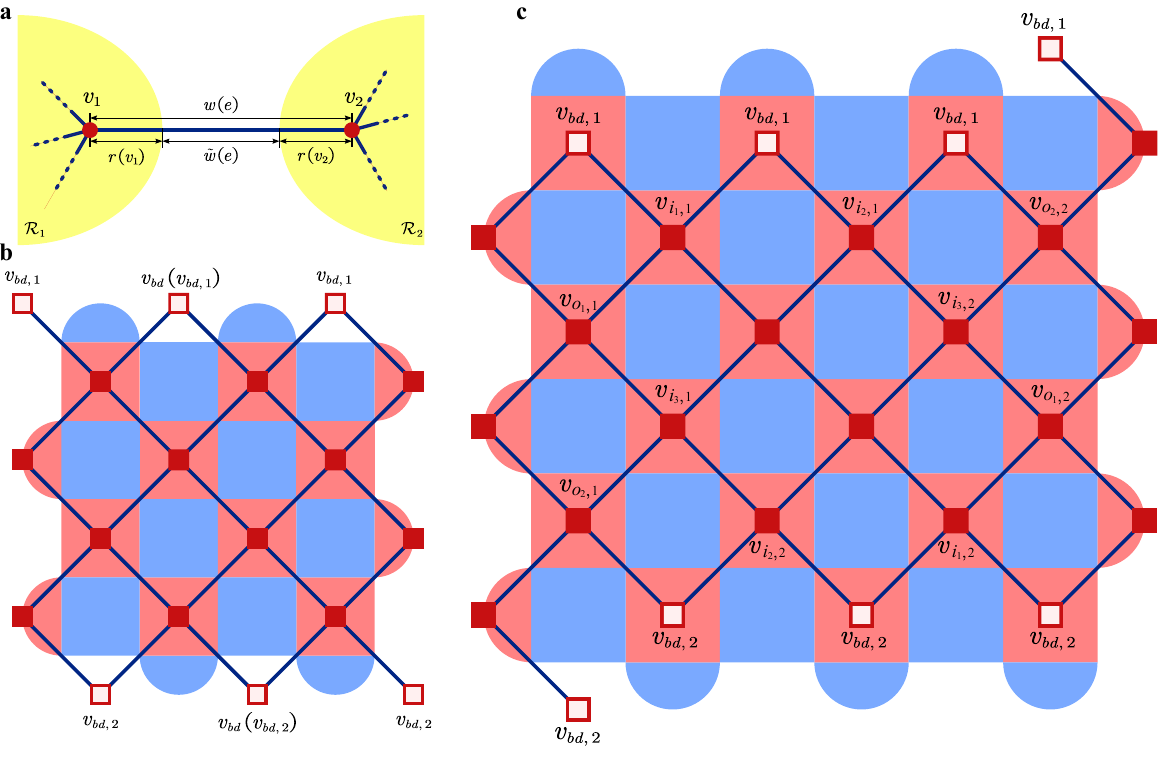}
    \caption{Obtaining soft outputs. (\textbf{a}) Obtaining post-decoding edge weights based on \textit{regions} at the end of a decoding process. Edge $e$ between vertex $v_1$ and $v_2$ with original weight $w(e)$ is reweighted to $\tilde{w}(e)$. \textit{Regions} $\mathcal{R}_1$ and $\mathcal{R}_2$ are colored by yellow. (\textbf{b}) Slice of $\tilde{\mathcal{G}}_{\mathrm{X}}$ on vertices (X detectors) generated at a single SE round. Stabilizers of the rotated surface code are plotted simply to indicate positions of the X detectors. X detectors are marked by smaller dark red squares and errors mechanisms (edges of $\tilde{\mathcal{G}}_{\mathrm{X}}$) by blue lines. A virtual boundary vertex $v_{bd}$ added to $\tilde{\mathcal{G}}_{\mathrm{X}}$ is connected to boundary edges and is then split into two distinct vertices $v_{bd,1}$ and $v_{bd,2}$ (hollow squares) at the upper and lower boundary respectively~\cite{meister_efficient_2024}. (\textbf{c}) Slice of $\tilde{\mathcal{G}}_{\mathrm{MSC},\mathrm{X}}$ on vertices (X detectors) generated immediately after the expansion step of our MSC-3 protocol. Similar to (\textbf{b}), there are two distinct virtual boundary vertices $v_{bd,1}$ and $v_{bd,2}$ (hollow squares) at the upper and lower boundary respectively. Each antipodal pair of vertices, such as $\{v_{i_a,0},v_{i_a,1}\}$ for $a\in\{1,2,3\}$ and $\{v_{o_b,0},v_{o_b,1}\}$ for $b\in\{1,2\}$, originally representing the same vertex (detector) is split into two distinct vertices. }
    \label{fig: soft output}
\end{figure}

Following~\cite{meister_efficient_2024,gidney_how_to_2024}, we first briefly review how the soft output~\cite{meister_efficient_2024} for decoding a rotated surface code memory can be obtained from the sparse blossom~\cite{higgott_sparse_2025} decoding process used in PyMatching. As we can decode on $\mathrm{X}$ detectors and $\mathrm{Z}$ detectors separately,  we restrict to decoding on the X detectors for now. Consider a decoding graph $\mathcal{G}_{\mathrm{X}}(\mathcal{V}_{\mathrm{X}},\mathcal{E},w)$ ($\mathcal{G}_{\mathrm{X}}$ for short) with vertex set $\mathcal{V}_{\mathrm{X}}$ (the set of all $\mathrm{X}$ detectors) and edge set $\mathcal{E}$, in which each edge $e$ may either be a normal edge connecting two vertices in $\mathcal{V}_{\mathrm{X}}$ or a boundary edge connected to only one vertex in $\mathcal{V}_{\mathrm{X}}$ and is assigned a non-negative weight $w(e)\in\mathbb{R}_{\geq 0}$.  The distance $D(u,v)$  between two vertices $u,v\in\mathcal{V}_{\mathrm{X}}$ is defined as the length of the shortest path in $\mathcal{G}_{\mathrm{X}}$ connecting them.  A \textit{combined graph fill region} (or \textit{region} for short) $\mathcal{R}$ is defined as a collection of vertices $V(\mathcal{R})\subset\mathcal{V}_{\mathrm{X}}$ equipped with a non-empty subset (named as the source subset), $S(\mathcal{R})\subset V(\mathcal{R})$, and a local radius map $r_{\mathcal{R}}:V(\mathcal{R})\to\mathbb{R}_{\geq 0}$, such that,
\begin{itemize}
    \item a vertex $u\in V(\mathcal{R})$ iff there exists $s\in S(\mathcal{R})$ such that $D(u,s)\leq r_{\mathcal{R}}(s)$, and, 
    \item for $\forall v\in V(\mathcal{R})$, \[r_{\mathcal{R}}(v)= \min_{s\in S(\mathcal{R})} \left( r_{\mathcal{R}}(s)- D(v,s)\right).\]
\end{itemize}
 Heuristically, the range of $\mathcal{R}$ is characterized by its local radius map. A vertex $v\in V(\mathcal{R})$ is said to be contained in $\mathcal{R}$ with local radius $r_{\mathcal{R}}(v)$ and, in particular, is said to be a source of $\mathcal{R}$ iff $v\in S(\mathcal{R})$.
  (Actually, in~\cite{higgott_sparse_2025}, only \textit{graph fill regions} are defined. The \textit{
  combined graph fill regions} defined here for our own convenience are combinations of \textit{graph fill regions}. More precisely, a \textit{combined graph fill region} is the combination of an \textit{active graph fill region}~\cite{higgott_sparse_2025} with all of its descendants.)  Consider a single decoding shot where we are given a set of triggered detectors $\mathcal{D} \subset \mathcal{V}_{\mathrm{X}}$.  The sparse blossom algorithm~\cite{higgott_sparse_2025} starts by initialize, for each $u\in\mathcal{D}$, a \textit{region} that only contains $u$ whose local radius is initialized to be $0$. Throughout the algorithm, \textit{regions} (their vertices, sources and local radius maps) are varied and they may merge into larger \textit{regions} or shatter into smaller \textit{regions}. Despite the evolution of \textit{regions}, their source subsets are mutually disjoint and only contain vertices in $\mathcal{D}$. Moreover, at any moment, any \textit{region} $\mathcal{R}$ contains an odd number of sources ($|S(\mathcal{R})|\equiv 1\mod 2$).  At the end of the decoding procedure, all regions are either matched to each other or to the boundary and will no longer be varied. Heuristically, one may assign each detection event a $+1$ charge (here charges are measured modulo $2$) and think of each region as a charged region containing an odd number of detection events, so that, at the end of the decoding process, charged regions are `neutralized' by either matching two regions together or matching a region to the boundary, which absorbs charges. Now, given the existing \textit{regions} at the end of the above decoding process, we can define for each edge $e$ a post-decoding weight $\tilde{w}(e)$ by erasing the part of its original weight covered by them. More precisely, for each endpoint $v\in\partial e$ ($\partial e$ denotes the set of endpoints of $e$), define its local radius $r(v)$ as $r_{\mathcal{R}}(v)$ if $v$ is contained in a region $\mathcal{R}$ and $0$ otherwise. Then $\tilde{w}(e)$ is defined as $\tilde{w}(e)=\max (0,w(e)-\sum_{v\in\partial{e}}r(v))$. (See Fig.~\ref{fig: soft output}(\textbf{a}) for an example where both endpoints of $e$ are contained in two different \textit{regions}.) Based on $\mathcal{G}_{\mathrm{X}}$, we define a new graph $\tilde{\mathcal{G}}_{\mathrm{X}}(\mathcal{V}_{\mathrm{X}}\sqcup \{v_{bd}\},\tilde{\mathcal{E}},\tilde{w})$ (or $\tilde{\mathcal{G}}_{\mathrm{X}}$ for short) with an additional virtual boundary vertex $v_{bd}$ and $\tilde{\mathcal{E}}$ is obtained from $\mathcal{E}$ by attaching each boundary edge in $\mathcal{E}$ to $v_{bd}$ (Fig.~\ref{fig: soft output}(\textbf{b})).  Denote the set of loops in $\tilde{\mathcal{G}}_{\mathrm{X}}$ that also flips a logical $\mathrm{X}$ observable as $\mathsf{L}(\tilde{\mathcal{G}}_{\mathrm{X}})$. The \textit{soft output} $\phi$ for the decoding shot above is defined as the minimal path length in $\mathsf{L}(\tilde{\mathcal{G}}_{\mathrm{X}})$~\cite{meister_efficient_2024}: \[\phi_{\mathrm{X}}:=\min_{\ell\in\mathsf{L}(\tilde{\mathcal{G}}_{\mathrm{X}})} \left(\sum_{e\in\ell}\tilde{w}(e)\right).\]
To evaluate $\phi_{\mathrm{X}}$, $v_{bd}$ is split into two distinct vertices $v_{bd,1}$ and $v_{bd,2}$ at the top and bottom boundary respectively (Fig.~\ref{fig: soft output}\textbf{b}). As a result, $\forall \ell\in\mathsf{L}(\tilde{\mathcal{G}})$ is now simply a path connecting $v_{bd,1}$ to $v_{bd,2}$ (Fig.~\ref{fig: soft output}(\textbf{b})). The \textit{soft output} $\phi_{\mathrm{X}}$, now as the distance between $v_{bd,1}$ and $v_{bd,2}$, can be efficiently computed via Dijkstra's algorithm~\cite{meister_efficient_2024}. 
Intuitively, suppose errors in a shot are sparse and do not cluster nor form long error chains, there would also be very few \textit{regions} after the decoding process and they would also have small radii. In this case, as the post-decoding edge weight would still be close to the original edge weight for almost all edges, the soft-output $\phi$ should still have a large value $\sim O(d)$ where $d$ is the code distance. In contrast, proliferation of errors may result in larger \textit{regions} and thus a smaller $\phi$. Thus, higher $\phi$ values indicate stronger decoding confidence. As shown in~\cite{meister_efficient_2024}, discarding all shots with soft outputs below a certain cutoff can significantly suppress the logical error rate and at the same time maintain a low discard rate.

\begin{figure}
    \centering
    \includegraphics[width=\columnwidth]{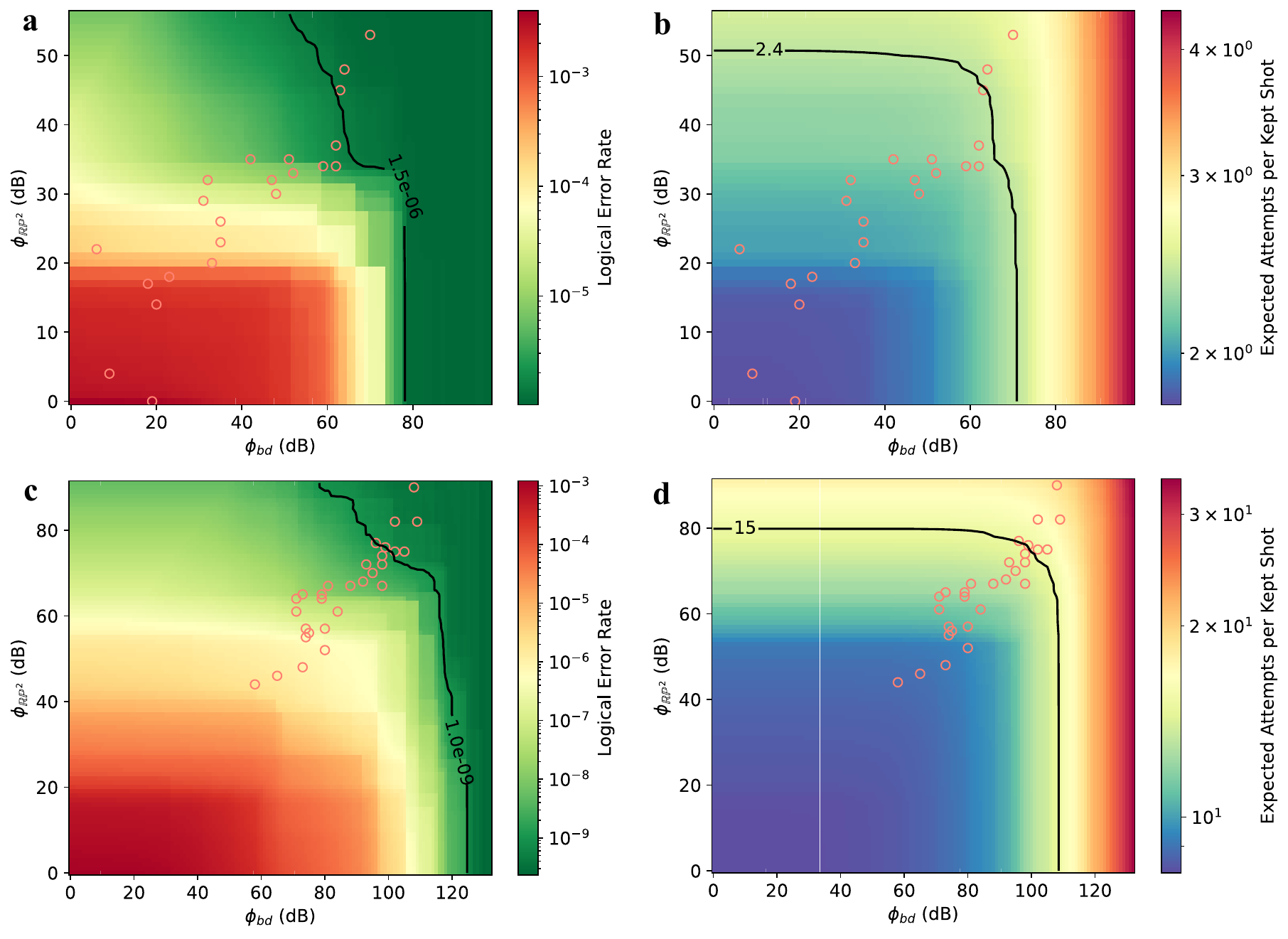}
    \caption{(\textbf{a}-\textbf{b}) Logical error rate and expected attempts per kept shot for our MSC-3 protocol at each soft-output cutoff. (\textbf{c}-\textbf{d}) Logical error rate and expected attempts per kept shot for our MSC-5 protocol at each soft-output cutoff. Cutoffs used in Fig.~\ref{fig: MSC results}(\textbf{b}) are marked by red circles.}
    \label{fig: soft_output cutoffs tradeoffs} 
\end{figure}

We now describe how a \textit{two-dimensional soft output} is extracted for decoding our MSC protocol. We focus on the MSC-3 protocol in the following. (Decoding the MSC-5 protocol is done similarly.) Consider decoding on X detectors first and denote the decoding graph as $\mathcal{G}_{\mathrm{MSC,X}}$. Similar to constructing $\tilde{\mathcal{G}}_{\mathrm{X}}$ from $\mathcal{G}_{\mathrm{
X}}$ above, for each decoding shot, we construct $\tilde{\mathcal{G}}_{\mathrm{MSC},\mathrm{X}}$ from $\mathcal{G}_{\mathrm{MSC},\mathrm{X}}$ by replacing edge weights with post-decoding edge weights and connecting boundary edges to a virtual boundary vertex $v_{bd}$, which is then split into two boundary vertices $v_{bd,1}$ and $v_{bd,2}$ at the top and bottom boundary respectively. Right after the expansion step, the X detectors and the slice of $\tilde{\mathcal{G}}_{\mathrm{MSC,X}}$ on them are illustrated in Fig.~\ref{fig: soft output}(\textbf{c}). Interestingly, as a result of our expansion circuit,  each antipodal pair $\{v_{i_a,0},v_{i_a,1}\}$ ($a\in\{1,2,3\}$) or $\{v_{o_b,0},v_{o_b,1}\}$ ($b\in\{1,2\}$) of vertices in Fig.~\ref{fig: soft output}(\textbf{c}) actually represents a single vertex (or equivalently a single X detector). Here, we also treat each of these antipodal pairs of vertices as distinct vertices. As a result, each path from $v_{i_{a},0}$ to $v_{i_a,1}$ for $\forall a\in\{1,2,3\}$ corresponds to a logical $\mathrm{Z}$ error. Similar to obtaining $\phi_{\mathrm{X}}$ above, we calculate via Dijkstra's algorithm the minimal distance $\phi_{\mathrm{X},a}$ between $v_{i_a,0}$ and $v_{i_a,1}$ for $\forall a\in\{1,2,3\}$ and also the minimum distance $\phi_{\mathrm{X},bd}$ between $v_{bd,0}$ and $v_{bd,1}$. We then define $\phi_{\mathrm{X},\mathbb{RP}^2}:=\min_{a\in\{1,2,3\}}\phi_{\mathrm{X},a}$.  Similarly, for decoding on Z detectors, we can obtain $\phi_{\mathrm{Z},\mathbb{RP}^2}$ and $\phi_{\mathrm{Z},bd}$. Finally, the \textit{two-dimensional soft output} is defined as
\begin{align}
    \phi_{\mathrm{MSC}} & := (\phi_{\mathbb{RP}^2},\phi_{bd}) \nonumber \\
    & := (\min\{\phi_{\mathrm{X},\mathbb{RP}^2},\phi_{\mathrm{Z},\mathbb{RP}^2}\}, \min\{\phi_{\mathrm{X},bd},\phi_{\mathrm{Z},bd}\}).
\end{align}

In our Monte Carlo simulations, $\phi_{\mathrm{MSC}}$ for each shot is measured in dB and is rounded to a tuple of integers. For a given soft-output cutoff $(\phi_1,\phi_2)\in\mathbb{Z}_{\geq0}^2$, we perform final-stage post-selection by only retaining shots whose soft outputs $(\phi_{\mathbb{RP}^2},\phi_{bd})$ satisfy $\phi_{\mathbb{RP}^2}\geq \phi_1$ and $\phi_{bd}\geq \phi_2$. The resulting logical error rate and the expected number of attempts (per kept shot) under each choice of $(\phi_1,\phi_2)$ are shown in Fig.~\ref{fig: soft_output cutoffs tradeoffs}(\textbf{a-b}) for our MSC-3 protocol and in Fig.~\ref{fig: soft_output cutoffs tradeoffs}(\textbf{c-d}) for our MSC-5 protocol. Around a certain discard rate (or an certain expected number of attempts per kept shot), we can select a cutoff with the lowest logical error rate. Such a cutoff is said to have almost optimal local trade-offs (between discard rate and logical error rate). Data points used in Fig.~\ref{fig: MSC results} correspond to cutoffs with almost optimal local trade-offs.  In particular the selected cutoffs used in Fig.~\ref{fig: MSC results}(\textbf{b}) are illustrated by red circles in Fig.~\ref{fig: soft_output cutoffs tradeoffs}.

\section{Discussion and outlook}\label{sec: discussion and outlook}

In this work, we designed a new MSC protocol that, according to Monte Carlo sampling results, requires nearly an order of magnitude lower space-time cost to achieve a logical error rate around $1.5\times 10^{-6}$ (for our MSC-3 version) or $1\times 10^{-9}$ (for our MSC-5 version) compared to the original MSC protocol~\cite{gidney_magic_2024}. The origin of the space-time cost reduction is that we replace the costly ``grafting'' process~\cite{gidney_magic_2024} with a more efficient ES process while also having a cultivation process with comparable efficiency to the one in the original MSC protocol. Furthermore, our MSC protocol would produce $\mathrm{T}$ states on regular rotated surface codes. Alternatively, the original MSC protocol can produce $\mathrm{T}$ states on either `grafted' codes (rotated surface codes deformed at a corner) or on regular color codes. (The latter requires color code patch growth and as of now does not admit efficient decoding with high performance~\cite{gidney_magic_2024}.) It would be interesting to develop fast color code decoders with better performance and then compare the efficiency of the original MSC protocol (that simply grows a smaller color code to a larger color code instead of ``grafting'') with ours and more generally compare the color code FTQC scheme with the surface code one.   

Our cultivation process (before the expand-then-stabilize process) has two extra steps (that allow morphing between $\mathbb{RP}^2$ codes with SRP codes) for each double-checking step and is currently slightly less efficient than the original cultivation process in~\cite{gidney_magic_2024}. However, the $\mathbb{RP}^2$ codes in our protocol have more compact SE circuits than the color codes used in~\cite{gidney_magic_2024}. The simplicity in our SE circuits also allows us to spend only two SE rounds (one SE round less than in~\cite{gidney_magic_2024}) to grow from a distance-3 code to a distance-5 code and stabilize. We expect our cultivation process can be further improved by designing more compact circuits (such as the double-checking circuits), which are now designed to be regular and conceptually simple. Moreover, for our cultivation process, we designed four new codes and made extensive use of their symmetries. We use the self-duality of the SRP-3 code and the SRP-5 code to design logical double-checking circuits on them. We use the fold-duality of the $\mathbb{RP}^2$-3 code and the $\mathbb{RP}^2$-5 code to design fast circuits that transform them to the SRP-3 code and the SRP-5 code respectively and then back. It is possible that MSC protocols based on other fold-dual and self-dual codes can be constructed to have even higher efficiency and/or accuracy than ours. It would also be interesting to explore applications of other types of ZX dualities or other code symmetries in preparing magic states.

\textit{Notes.} The codes and circuits used in this work is made publicly available at~\cite{zihan_culti}. Z. C. would like to thank Nadine Meister for discussions on the \textit{soft-output} method in~\cite{meister_efficient_2024}.

\appendix

\section{Extra details}\label{appsec: extra details}

In this appendix, we describe details of our simulation including:
\begin{itemize}
    \item circuit noise model,
    \item how certain circuits are run concurrently, and
    \item how space-time cost is computed. 
\end{itemize}

\begin{figure}
    \centering
    \includegraphics[width=\columnwidth]{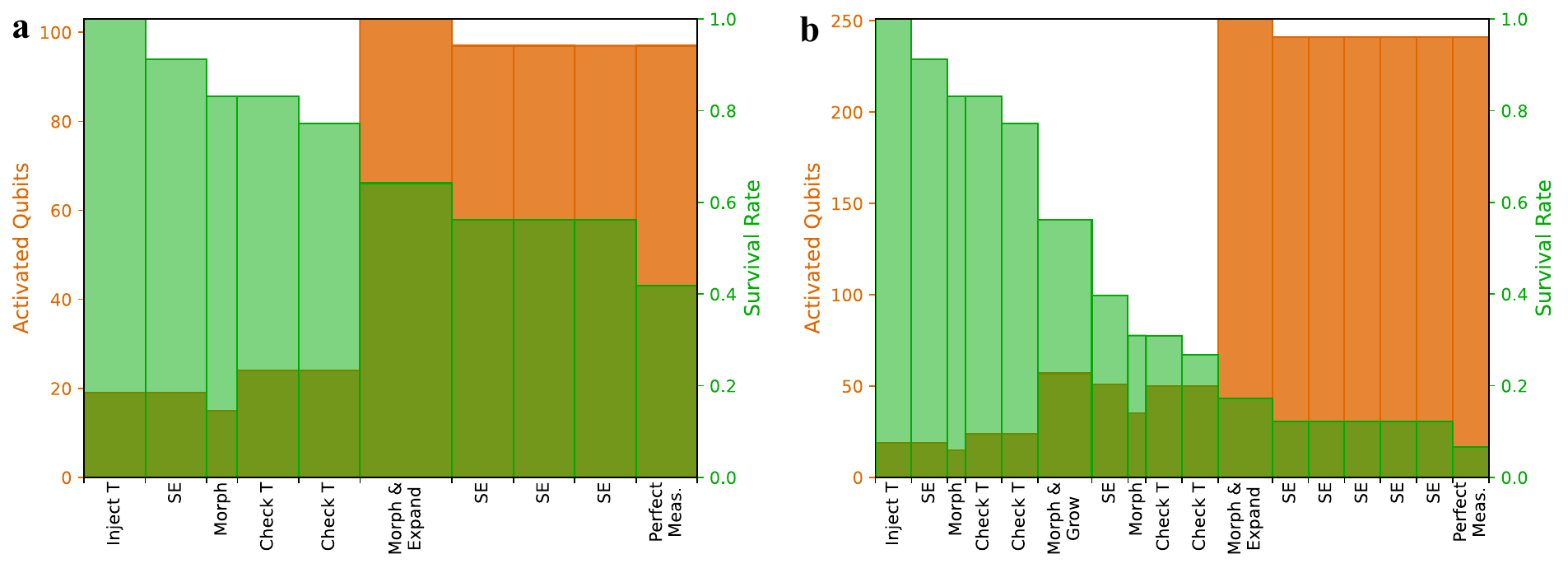}
    \caption{(\textbf{a}) Activated qubits and survival rate at each round of the MSC-3 protocol. (\textbf{b}) Activated qubits and survival rate at each round of the MSC-5 protocol. Each bin corresponds to a single round and has width proportional to the round length.}
    \label{fig: MSC at each round}
\end{figure}

We use the uniform depolarizing circuit noise model as in~\cite{gidney_magic_2024}. More specifically, our circuit error configurations are listed in the following:
\begin{itemize}
\item [i.] each single-qubit gate (including idling) is followed by a single-qubit depolarization channel with noise strength $p$,
\item [ii.] each two-qubit gate is followed by a two-qubit depolarization channel with noise strength $p$,
\item [iii.] initialization in $\pm\mathrm{Z}$ bases is followed by a single-qubit $\mathrm{X}$ error with probability $p$,
\item [iv.] initialization in $\pm\mathrm{X}$ bases is followed by a single-qubit $\mathrm{Z}$ error with probability $p$, and
\item [v.] each single-qubit measurement result is flipped with probability $p$.
\end{itemize}
As mentioned in Sec.~\ref{sec: results}, we sampled on a Clifford variant of our MSC protocol by replacing $\mathrm{T}$ ($\mathrm{T}^{\dagger}$) gates with $\mathrm{S}$ ($\mathrm{S}^{\dagger}$) gates following~\cite{gidney_magic_2024}. The Clifford variant would produce a logical $\mathrm{S}|+\rangle$ state instead of a $\mathrm{T}$ state. We perform accurate state vector simulation~\cite{gidney_magic_2024} on the cultivation process of our MSC-3 protocol under no circuit noise to verify its correctness. In~\cite{gidney_magic_2024}, state vector simulation of cultivation process (with $d_{\text{color}}=3$) under circuit noise is also performed. It is observed in~\cite{gidney_magic_2024} that, at circuit noise rate $p=0.001$, the error rate for cultivating a $\mathrm{T}$ state (estimated via state vector simulation) on a $d_{\text{color}}=3$ color code is about two times larger than the error rate for cultivating an $\mathrm{S}|+\rangle$ state (estimated via Monte Carlo sampling). In our work, however, state vector simulation is not efficient enough to provide enough samples to estimate the logical error rate under circuit noise rate $p=0.001$ since our cultivation process (of the MSC-3 protocol) has a larger spatial footprint ($24$ qubits) compared to the cultivation process with $d_{\text{color}}=3$ (using $15$ qubits) in the original MSC protocol. It would be interesting to adopt more efficient sampling methods~\cite{bravyi_simulation_2019,kissinger_classical_2022} to sample our cultivation process. We leave such implementation to future work. 

As $\mathbb{RP}^2$-growth circuits and expansion circuits both require preparing Bell pairs, we run them concurrently with morphing circuits that transform SRP codes back to $\mathbb{RP}^2$ codes. Take the expansion circuit in our MSC-3 protocol as an example. Morphing an SRP-3 code back to an $\mathbb{RP}^2$-3 code is implemented by running three copies of the decoding circuit (containing $4$ steps as in Fig.~\ref{fig: encoding and decoding circuits of small self dual codes}(\textbf{b})) for the $[[4,2,2]]$ code.   We see that since the third step of the decoding circuit, data qubits of the $\mathbb{RP}^2$-3 are no longer used in the decoding circuit and are available for the SE part of the expansion circuit. Thus, to prepare the Bell pairs, as a preliminary requirement for the the expansion circuit, we first insert an initialization step that initialize all ancilla and data qubits for the expansion circuit between the first step and the second step of the decoding circuit. Then, at the second step of the decoding circuit, we concurrently run 8 extra CNOT gates to prepare the Bell pairs in Fig.~\ref{fig: circuits for patch growth}(\textbf{c}). At the third step of the decoding circuit, we can concurrently run the first round of SE in the expansion circuit. The rest of the SE circuit is then run after the last step the decoding circuit. In this way, an $\mathbb{RP}^2$ growth step or an expansion step is naturally merged with the preceding (backwards) morphing stage. 

Following~\cite{gidney_magic_2024}, we parse our MSC protocol into circuit rounds as in Fig.~\ref{fig: MSC at each round}. These rounds are naturally assigned different round lengths according to the number of resets and measurements contained in each round. More specifically:
\begin{itemize}
\item morphing from an $\mathbb{RP}^2$ code to an SRP code is considered as a single round with length 0.5,
\item  morphing from an SRP code to an $\mathbb{RP}^2$ code and growing (expanding) the $\mathbb{RP}^2$ code are bundled together as a round with length 1.5, and
\item  other rounds all have length 1.
\end{itemize}
 For each round, the expected space-time cost is defined as the product of the number of active qubits for this round, the survival rate at the beginning of this round and the round length. Similar to~\cite{gidney_magic_2024}, the space-time cost for our MSC protocol is defined as the sum of expected space-time cost for each round divided by the survival rate at the end of the protocol.

\bibliography{bibliography}

\end{document}